\pdfoutput=1
\documentclass[11pt]{article}

\usepackage[utf8]{inputenc} 
\usepackage[T1]{fontenc}    
\usepackage{hyperref}       
\usepackage{url}            
\usepackage{booktabs}       
\usepackage{amsfonts}       
\usepackage{nicefrac}       
\usepackage{microtype}      
\usepackage{color}
\usepackage{amsmath,amssymb}
\usepackage{amscd,amsthm}

\usepackage{graphicx} 

\usepackage[
pass,
]{geometry}

\newgeometry{top=1in, bottom=1in, left=1in, right=1in}

\newfont{\mycrnotice}{ptmr8t at 7pt}
\newfont{\myconfname}{ptmri8t at 7pt}

\newcommand{\A}{\mathcal{A}}
\newcommand{\Alg}{\mathfrak{A}}

\newcommand{\T}{\mathfrak{T}}
\newcommand{\Str}{\mathfrak{S}}

\newcommand{\astr}{\mathbf{a}}
\newcommand{\bstr}{\mathbf{b}}

\newcommand{\SReg}{\mathrm{SReg}}
\newcommand{\Sur}{\mathrm{Sur}}

\newcommand{\Rev}{\mathrm{Rev}}
\newcommand{\SRev}{\mathrm{SRev}}

\newcommand{\n}{\mathfrak{n}}
\newcommand{\m}{\mathfrak{m}}
\newcommand{\estr}{\mathfrak{e}}
\newcommand{\rt}{\mathfrak{r}}
\newcommand{\lf}{\mathfrak{l}}

\newcommand{\Ind}{\mathbb{I}}
\newcommand{\fa}{\:\forall}

\newcommand{\g}{\gamma}

\newcommand{\argmax}{\mathop{\mathrm{argmax}}}

\newcommand{\Expect}{\mathbb{E}}
\newcommand{\Prob}{\mathbb{P}}

\newcommand{\gb}{{\boldsymbol\g}}
\newcommand{\nub}{{\boldsymbol\nu}}

\newcommand{\vb}{{\mathbf{v}}}
\newcommand{\wb}{{\mathbf{w}}}
\newcommand{\pb}{{\mathbf{p}}}
\newcommand{\rb}{{\mathbf{r}}}

\newcommand{\nO}{\mathtt{0}}
\newcommand{\nI}{\mathtt{1}}
\newcommand{\N}{\mathfrak{N}}

\newcommand{\ttS}{\mathtt{S}}
\newcommand{\ttB}{\mathtt{B}}

\theoremstyle{plain}
\newtheorem{theorem}{Theorem}
\newtheorem{lemma}{Lemma}

\newtheorem{corollary}{Corollary}
\newtheorem{proposition}{Proposition}
\theoremstyle{definition}
\newtheorem{remark}{Remark}

\newtheorem{definition}{Definition}

\begin{document}
\title{Optimal Pricing in Repeated Posted-Price Auctions}  
\author{
	 {\bf Arsenii Vanunts}\\ Yandex\thanks{16, Leo Tolstoy St., Moscow, Russia, 119021 (www.yandex.com)}, MSU\thanks{Lomonosov Moscow State University, Faculty of Mechanics and Mathematics; GSP-1, 1 Leninskiye Gory, Main Building, Moscow, Russia, 119991} \\ {\tt\small avanunts@yandex.ru} \\
	  {\bf Alexey Drutsa}\\ Yandex\thanks{16, Leo Tolstoy St., Moscow, Russia, 119021 (www.yandex.com)}, MSU\thanks{Lomonosov Moscow State University, Faculty of Mechanics and Mathematics; GSP-1, 1 Leninskiye Gory, Main Building, Moscow, Russia, 119991} \\ {\tt\small adrutsa@yandex.ru}
}
\date{19 March 2018}

\maketitle

\begin{abstract}
We study revenue optimization pricing algorithms for repeated posted-price auctions where a seller interacts with a single strategic buyer that holds a fixed private valuation. 
We show that, in the case when both the seller and the buyer have the same discounting in their cumulative utilities (revenue and surplus), there exist two optimal algorithms. The first one constantly offers the Myerson price, while the second pricing proposes a ``big deal": pay for all goods in advance (at the first round) or get nothing.
However, when there is an imbalance between the seller and the buyer in the patience to wait for utility, we find that the constant pricing, surprisingly, is no longer optimal. 
First, it is outperformed by the  pricing algorithm ``big deal", when the seller's discount rate is lower than the one of the buyer.
Second, in the inverse case of a less patient buyer, we reduce the problem of finding an optimal algorithm to a multidimensional optimization problem (a multivariate analogue of the functional used to determine Myerson's price) that does not admit a closed form solution in general, but  can be solved by numerical optimization techniques (e.g., gradient ones).
We provide extensive analysis of numerically found optimal algorithms to demonstrate that they are non-trivial, may be non-consistent, and generate larger expected revenue than the constant pricing with the Myerson price.
\end{abstract}

\section{Introduction}

Revenue maximization in online advertising is an important development direction of leading Internet companies (like real-time ad exchanges~\cite{2015-ManagSci-Balseiro}, search engines~\cite{2013-IJCAI-He}, and social networks), in which a large part of ad inventory is sold via widely applicable second price auctions~\cite{2013-IJCAI-He,2014-ICML-Mohri}, including the generalizations GSP~\cite{2014-ECRA-Sun} and VCG~\cite{1981-MOR-Myerson}. 
The optimization of revenue in these auctions is mostly controlled by means of  reserve prices, whose proper setting is studied both by game-theoretical methods~\cite{1981-MOR-Myerson,2009-Book-Krishna} and by machine learning approaches~\cite{2007-Book-Nisan,2013-SODA-Cesa-Bianchi,2014-ECRA-Sun,2014-ICML-Mohri,2017-NIPS-Medina,2017-WWW-Drutsa}. 
A large number of online auctions in, for example, ad exchanges involve only a single buyer~\cite{2013-NIPS-Amin,2014-NIPS-Mohri,2014-NIPS-Amin,2017-WWW-Drutsa}, and, in this case, a second-price auction with reserve reduces to a \emph{posted-price auction}~\cite{2003-FOCS-Kleinberg} where the seller sets a reserve price for a good (e.g., an advertisement
space) and the buyer decides whether to accept or reject it (i.e., to bid above or below the price).

In our study, we focus on a scenario in which the seller \emph{repeatedly} interacts through a posted-price mechanism with the \emph{same} strategic buyer that holds a \emph{fixed} private valuation for a good and seeks to maximize his cumulative  surplus.
At each round of this game, the seller is able to chose the price based on previous  decisions of the buyer: he applies a deterministic online learning algorithm announced to the buyer in advance~\cite{2014-NIPS-Mohri}.
While previous studies on this scenario~\cite{2013-NIPS-Amin,2014-NIPS-Mohri,2017-WWW-Drutsa} provide the seller with pricing algorithms that guarantee \emph{lower bounds} on his cumulative revenue  for any buyer valuation (via worst-case strategic regret minimization), we search for pricing algorithms that \emph{exactly} maximize the \emph{expectation} of the seller's cumulative revenue over a given distribution of buyer valuations.
The cumulative utilities (surplus for the buyer and revenue for the seller) are considered as discounted sums of corresponding instant utilities gained at each round, what allows us to cover a wide range of games (including the ones with infinite number of rounds and finite games without discounting).

We start our study from addressing the case when both the seller and the buyer have the same discount.
We show that the constant pricing algorithm with the Myerson price $p^\ast=\argmax_pH_D(p),$ where $ H_D(p) = p \cdot \Prob_{V\sim D}[V\ge p]$ and $D$ is the valuation distribution,  maximizes our optimization objective (see Theorem~\ref{maintheorem}). 
This result tells us that any dynamic learning of prices based on previous  decisions of the buyer can \emph{not} increase the expected cumulative revenue of the seller with respect to a much simpler approach that offers the optimal constant price over all rounds.
Further we also show that the above mentioned optimal pricing is not unique. 
Namely, there exists an optimal pricing algorithm (referred to as ``big deal") that proposes the following choice to the buyer: pay a large price at the first round and get all goods in the subsequent rounds for free, otherwise get nothing (see Prop.~\ref{prop_non_uniqu}).
The same discount for both participants of the game assumes that we do not give any advantage to each of them over the other one. However, in many real applications, there exists an imbalance between the sides in the patience to wait for utility. This asymmetry is often modeled by different discounts  for them~\cite{2013-NIPS-Amin,2014-NIPS-Amin,2014-NIPS-Mohri}.
In our work, we address both the case of less patient seller and the case of less patient buyer.

First, in the case when the buyer's discount rate is larger than the seller's one, we find that the algorithm ``big deal" with a specific price at the first round can still be effectively applied by the seller (i.e., with optimal outcome). Namely, it allows the seller to  ``accumulate" all his revenue at the first round and, in this way, to avoid the uncomfortable discounting in the future rounds; this discount makes the constant algorithm with Myerson's price suboptimal (see Sec.~\ref{sec_NonEqualDisc_BgS}).
Second, in the inverse case, when the buyer's discount rate is lower than the seller's one, the optimization problem becomes surprisingly more complicated. In this case, we reduce it to the optimization of a bilinear form in $\vb = \{v_j\}_j$ and $\{\Prob_{V\sim D}[V \ge v_i]\}_j$   (see Theorem~\ref{th_problem_equivalence}). 
This functional constitutes a multivariate analogue of the one-dimensional function  $H_D(p) $ widely used in static auctions to find the optimal pricing. 
Our reduction does not admit a closed form solution in general, but  allows to find the optimal algorithm by means of state-of-the-art numerical optimization techniques (e.g., gradient ones).
In contrast to the previous cases, the optimal algorithm in this case of less patient buyer is non-trivial and its prices depend on both the valuation distribution and the discounts. 
Finally, we numerically solve the above mentioned reduced problem for a series of representative discounts and analyze properties of  the obtained optimal algorithms (see Sec.~\ref{sec_NonEqualDisc_BlS}).  In this way, we show, in particular, that an optimal algorithm may  be non-consistent\footnote{A consistent algorithm never sets prices lower (higher) than earlier accepted (rejected, resp.) ones.} and provides revenue larger than the constant algorithm with Myerson's price.


The most important conclusion consists in the following. 
Only in the case of equal discounts, the seller cannot advantageously use the ability to change prices in dynamic fashion (i.e., to learn them) w.r.t.\ the static approach. But, both in the case when the seller is far more ready to wait for revenue than the buyer, and, more surprisingly, in the inverse case, the seller can boost his revenue w.r.t.\ the one obtained by the optimal constant algorithm. 
Overall, the above described thorough study of optimal pricing algorithms for repeated auctions with different discounts constitutes the main contribution of our work.
The ideas behind our techniques of theoretical analysis are simple and, to the best of our knowledge, novel; they might thus be used for future foundations of repeated auctions, e.g., the ones with multiple buyers.




\section{Preliminaries, problem statement and related work}
\label{sec_Prelim}
\subsection{Setup of repeated posted-price auctions}
\label{subsec_Setup}
We consider the following standard mechanism of \emph{repeated posted-price auctions}~\cite{2013-NIPS-Amin,2014-NIPS-Mohri,2016-SSRN-Chen,2017-WWW-Drutsa,2017-ArXiV-Drutsa}.
The seller repeatedly proposes goods (e.g., advertisement spaces) to a single buyer over a sequence of rounds (one good per round). 
The buyer holds a \emph{fixed private valuation} $v \in [0; +\infty)$ for a good, i.e., the valuation $v$ is unknown to the seller and is equal for goods offered in all rounds. 
At each round $t \in \mathbb{N}$, the seller offers a price $p_t$ for a good, and the buyer makes his allocation decision $a_t \in \{0, 1\}$: to buy the currently offered good ($a_t = 1$), or not ($a_t = 0$). 
In our setting, the seller's price $p_t, t\in\mathbb{N},$ depends on the previous answers $a_1, .., a_{t - 1}$ of the buyer (a.k.a.\ the history up to the round $t$), i.e., the seller uses \emph{a pricing algorithm} $\A$ to set prices in the deterministic online learning manner~\cite{2013-NIPS-Amin,2014-NIPS-Mohri,2017-WWW-Drutsa}. 
The sequence of the buyer's answers is denoted by $\astr = \{a_t\}_{t = 1}^\infty$ and is referred to as \emph{a buyer strategy}. 

Hence, given an algorithm $\A$ and a strategy $\astr$, the price sequence $\{p_t\}_{t = 1}^\infty$ is uniquely determined. 
The \emph{instant surplus} $a_t (v - p_t)$ and the \emph{instant revenue} $a_t p_t$ are thus gained by the buyer and the seller, respectively, at each round $t\in\mathbb{N}$. 
An instant surplus (or revenue) obtained in different rounds may contribute differently to the total (cumulative) profit of the buyer (or the seller, respectively).
We model this by discount factors $\g^\ttB_t$ and $\g^\ttS_t$ at each round $t\in\mathbb{N}$ and get \emph{the total discounted surplus} and \emph{the total discounted revenue} of the following form:
\vspace{-0.2cm}
\begin{equation}
	\label{eq_Sur_Rev}
	\Sur_{\gb^\ttB}(\A, v, \astr) := \sum_{t = 1}^{\infty}\g^\ttB_t a_t (v - p_t)  \quad \text{and}\quad  \Rev_{\gb^\ttS}(\A, \astr) := \sum_{t = 1}^{\infty} \g^\ttS_t a_t p_t , \quad \hbox{respectively.}
\vspace{-0.2cm}
\end{equation}
We assume that \emph{the discount sequences} $\gb^\ttB = \{\g^\ttB_t\}_{t = 1}^\infty$ and $\gb^\ttS = \{\g^\ttS_t\}_{t = 1}^\infty$ are non-negative, $\g^\ttB_{t},\g^\ttS_{t}\ge0$, $\fa t\in\mathbb{N}$, and the series converges, $\Gamma^\ttB\!\!:=\!\!\sum_{t = 1}^{\infty} \g^\ttB_t, \Gamma^\ttS\!\!:=\!\!\sum_{t = 1}^{\infty} \g^\ttS_t < \infty$. We also assume that there are no zeros between positive numbers in the sequences $\gb^\ttB$ and $\gb^\ttS$.
Note that discounts allow us to consider a general setting, which covers a wide range of cases including finite games without discounting (i.e., $\g^\ttB_t = \g^\ttS_t=\Ind_{\{t\le T\}}$\footnote{$\Ind_{B}$ denotes the indicator of the condition $B$, i.e.,  $\Ind_{B} = 1$, when $B$ holds, and $0$, otherwise.} for some horizon $T\in\mathbb{N}$) and  infinite games with 
discount rates that decrease geometrically (i.e., $\g^\ttB_t = \g^\ttS_t=\g^{t-1}$ for some $\g\in(0,1)$)~\cite{2013-NIPS-Amin}.

Both the seller and the buyer may have the same discount ($\gb^\ttB_t = \gb^\ttS_t$), which is a reasonable assumption since it does not give any privilege to each party over the other one. For instance, money inflation, a common interpretation of the discount factor, affects the preferences of both participants for current
gains versus future ones equally.
The case when the discounts are different ($\gb^\ttB_t \neq \gb^\ttS_t$) is important for real applications as well~\cite{2013-NIPS-Amin}.
The discounting can also be considered as a model for uncertainty of the participants about the total number of rounds of their interaction (i.e., the factor $\g_t$ is a priori probability that repeated auctions will last exactly $t$ rounds).


Following a standard assumption in mechanism design, which matches the practice in ad exchanges~\cite{2014-NIPS-Mohri}, the pricing algorithm $\A$, used by the seller, \emph{is announced to the buyer in advance}~\cite{2013-NIPS-Amin,2017-WWW-Drutsa}. In this case, the buyer is able to act strategically against this algorithm, i.e., to chose \emph{the optimal strategy} $\astr^{Opt}(\A, v, \gb^\ttB)$ in the set of all possible strategies $\Str := \{0, 1\}^\mathbb{N}$, i.e., $ \astr^{Opt}(\A, v, \gb^\ttB) = \argmax_{\astr \in \Str} \Sur_{\gb^\ttB}(\A, v, \astr)$\footnote{\label{fn_SRev}We show existence of the maximum in Appendix~\ref{app_subsec_optstrat_exist}. If there is a tie, i.e., more than one optimal strategy, the buyer selects one of them arbitrary (as in~\cite{1981-MOR-Myerson,2009-Book-Krishna}).},
This leads us to the definition of the \emph{strategic revenue} of the pricing algorithm $\A$, which faces the strategic buyer with a valuation $v\in[0,\infty)$: 
\begin{equation}
	\label{eq_SRev}
	\SRev_{\gb^\ttS,\gb^\ttB}(\A, v): = \Rev_{\gb^\ttS}(\A, \astr^{Opt}(\A, v, \gb^\ttB)).
\end{equation}

\subsection{Notation and auxiliary definitions}
\label{subsec_Notations}
Following~\cite{2003-FOCS-Kleinberg,2014-NIPS-Mohri,2017-WWW-Drutsa}, we associate a deterministic pricing algorithm with a complete infinite binary tree $\T$ in which each vertex is labeled with a price. The algorithm offers the price from a current node (starting from the root) and moves to the left  (right) child of the node if the buyer answers $a_t=0$ ($=1$, respectively).
Clearly, buyer decisions at rounds $1,..,t$ encode bijectively paths from the root to tree nodes and, thus, nodes as well. Hence, we apply short notations for the nodes by means of the dictionary of finite strings $\N := \{\nO, \nI\}^\ast$: the root is the empty string $\estr$, its left child is $\nO$, the right one is $\nI$, the right child of $\nO$ is $\nO\nI$, etc. (e.g., $\nO^k$ denotes the string of $k$ zeros).
Similarly, we denote buyer strategies by infinite strings from the alphabet $\{\nO, \nI\}$\footnote{We purposely use different outline of the numbers \emph{zero} and \emph{one} to distinguish their use in numerical expressions (as $0$, $1$) and their use in strings that encode nodes or strategies (as elements of the alphabet $\{\nO, \nI\}$).} to save space (e.g., the buyer that follows $\nI\nO^\infty$ accepts the price  at the first round, $a_1=1$, and rejects all remaining ones, $a_t=0,t>1$).
Overall, the set of pricing algorithms $\Alg$ is equivalent to the set of mappings from the nodes $\N$ to $[0; +\infty)$, and we use thus  them interchangeably: $\Alg = [0; +\infty)^\N$. The price of an algorithm $\A\in\Alg$ offered at a node $\n\in\N$ is denoted by $\A(\n)$.

\subsection{Problem statement}
Let possible buyer valuations be distributed on $[0,+\infty)$ according to some distribution $D$, i.e., the buyer valuation $v$ (fixed over all rounds) is a realization of a random variable $V\sim D$. 
Following a standard assumption in classical auction theory~\cite{2007-Book-Nisan,2009-Book-Krishna}, \emph{the valuation distribution} $D$ is known by the seller.
We also assume that the distribution $D$ has finite expectation, i.e., $\Expect_{V\sim D}[V]<\infty$, and is continuous; these assumptions are standard in auction theory as well~\cite{1981-MOR-Myerson,2009-Book-Krishna}.
So, we consider the problem of finding a pricing algorithm  $\A^\ast \in \Alg$ that maximizes the expected strategic revenue\footnote{Note that, in repeated auctions, revenue is usually compared to the one that would have been earned by offering the buyer's valuation $v$ if it was known in advance to the seller, resulting in the notion of the strategic regret $\SReg_{\gb^\ttS,\gb^\ttB}(\A, v):= \Gamma^\ttS v \!-\!  \SRev_{\gb^\ttS,\gb^\ttB}(\A, v)$. Regret is a powerful instrument to obtain lower bounds on revenue~\cite{2003-FOCS-Kleinberg,2013-NIPS-Amin,2017-WWW-Drutsa}, but, in our setup, minimization of the expected strategic regret is equivalent to our problem.}: $\Expect_{V\sim D} [ \SRev_{\gb^\ttS,\gb^\ttB}(\A, V) ] \rightarrow \max$.

From a game-theoretic view,  we consider a two-player non-zero sum repeated game with incomplete information and unlimited supply in which the seller commits to the pricing (since he announces the algorithm before the auctions take place).
An attentive reader may also note that, due to the commitment and the presence of only one buyer, our setting can be formalized as a two stage game. 
The common knowledge here are the discounts $\gb^\ttB $, $\gb^\ttS $,  and the prior distribution $D$ of the private valuation $V$, while the realization $v$ of $V$ is known only by the buyer. 
At the first stage, the seller picks a pricing algorithm $\A\in\Alg$, his choice is announced to the buyer; 
at the second stage, the buyer picks a buyer strategy $\astr\in\Str$. 
The buyer's utility is the surplus and the seller's one is the expected revenue (see Eq.~(\ref{eq_Sur_Rev})). 
Thus, if some pricing $\A^*\in\Alg$ is a solution to our problem, then the pair $(\A^*,\astr^{Opt}(\A^*,v,\gb^\ttB))$ will be an equilibrium of above described game.



\begin{remark}
	\label{remark_first_discount}
	Note that both an optimal buyer strategy and an optimal algorithm will remain optimal, if the discount $\gb^\ttB$ or $\gb^\ttS$ is multiplied by any positive constant. Hence, from here on in our paper we assume w.l.o.g. that $\gamma^\ttB_1 = 1$ and $\gamma^\ttS_1 = 1$.
\end{remark}

\subsection{Related work}
\label{subsec_RelWork}
Optimization of seller revenue  in auctions was generally reduced to a selection of proper reserve prices for buyers\footnote{Of course, there are other options to optimize revenue like quality scores for advertisements in ad auctions~\cite{2013-IJCAI-He}, but they are significantly less popular. And, surely, revenue optimization was also considered in other contexts such as trade-offs between  auction stakeholders~\cite{2014-WWW-Goel} or between auction properties (e.g., simplicity, expressivity~\cite{2015-NIPS-Morgenstern}, and revenue monotonicity~\cite{2014-WWW-Goel}). } (e.g., in VCG~\cite{1981-MOR-Myerson}, GSP~\cite{2014-ECRA-Sun}, and other auctions~\cite{2016-WWW-Paes}).
In such setups, these prices usually depend on distributions of buyer bids or valuations~\cite{1981-MOR-Myerson} and was in turn estimated by machine learning techniques~\cite{2013-IJCAI-He,2014-ECRA-Sun,2016-WWW-Paes}, while alternative approaches learned reserve prices directly~\cite{2014-ICML-Mohri,2017-NIPS-Medina}.
In contrast to these works, we consider an online deterministic learning framework for repeated auctions.

Revenue optimization for repeated auctions was mainly concentrated on algorithmic reserve prices, that are updated in online fashion over time, and was also known as dynamic pricing, see the extensive survey~\cite{2015-SORMS-den-Boer} on this field.
Oh the one hand, dynamic pricing was studied under game-theoretic view in context of different aspects such as
budget constraints~\cite{2015-ManagSci-Balseiro,2016-EC-Balseiro},
mean field equilibria~\cite{2011-ECOMexch-Iyer,2015-ManagSci-Balseiro}, 
strategic buyer behavior~\cite{2015-EC-Chen,leme2012sequential},
multi-period contracts~\cite{1985-EL-Besanko}, etc.
A series of studies~\cite{1993-JET-Schmidt,2015-SODA-Devanur,2017-EC-Immorlica} close to ours considered repeated sales where the seller does not commit for its pricing policy (in contrast to our setting), what required thus special approaches (such as the concept of perfect Bayesian equilibrium) to address the revenue optimization problem. That studies showed that the seller earns less in settings without commitment than with it.
Another line of works like~\cite{pavan2014dynamic,2013-OR-Kakade} studied auction environment settings of a general form and was aimed to find revenue optimal mechanisms that are incentive compatible (truthful). In contrast to these studies, we consider a specific mechanism of repeated posted-price auctions and do not require its truthfulness (e.g., the algorithms in Sec.~\ref{subsec_finite_game_studies} and~\ref{subsec_infinite_game_studies}).
Finally, our work can be considered as further development of classical auction theory~\cite{2007-Book-Nisan,2009-Book-Krishna}: in particular, in the case of a more patient seller, to address the optimal pricing problem we derive a multidimensional optimization functional, defined in Eq.~(\ref{prop_problem_reduction_eq_2}), which is a multivariate analogue of the classical one, $p \cdot \Prob_{V\sim D}[V\ge p]$, used to determine the optimal reserve price in static auctions.
Overall, the optimal pricing in our scenario of repeated posted-price auctions with different discounts for the seller and the buyer, to the best of our knowledge, was never considered in existing studies, and we believe that the key ideas behind our analysis may be used for future foundation on repeated auctions.

Oh the other hand,  revenue optimization in dynamic pricing was considered from algorithmic and learning approaches:
as bandit problems~\cite{2011-COLT-Amin,2015-NIPS-Zoghi,2015-NIPS-Lin}
(e.g., UCB-like pricing~\cite{2015-TEC-Babaioff}, bandit feedback models~\cite{2016-JMLR-Weed});
from the buyer side (valuation learning~\cite{2011-ECOMexch-Iyer,2016-JMLR-Weed}, competition between buyers and optimal bidding~\cite{2014-WWW-Hummel,2016-JMLR-Weed}, interaction with several sellers~\cite{2016-ICML-Heidari}, etc.);
from the seller side against several buyers~\cite{2013-SODA-Cesa-Bianchi,2017-SSRN-Kanoria,2016-EC-Roughgarden,2016-NIPS-Feldman};
and a single buyer with stochastic valuation (myopic~\cite{2003-FOCS-Kleinberg,2016-SODA-Chawla} and strategic buyers~\cite{2013-NIPS-Amin,2014-NIPS-Amin,2014-NIPS-Mohri,2016-SSRN-Chen}, feature-based pricing~\cite{2014-NIPS-Amin,2016-EC-Cohen}, limited supply~\cite{2015-TEC-Babaioff}).

The most relevant studies from these works on online learning are \cite{2013-NIPS-Amin,2014-NIPS-Mohri,2017-WWW-Drutsa,2017-ArXiV-Drutsa}, where our scenario of the strategic buyer with a fixed private valuation is considered. 
Amin et al.~\cite{2013-NIPS-Amin} proposed to seek for algorithms that have the lowest possible upper bound on the strategic regret for the \emph{worst case} buyer valuation, i.e., $\sup_{v\in[0,1]}[\SReg_{\gb^\ttS,\gb^\ttB}(\A, v, \astr)]\le O(f(T))$, where $T$ is the finite game horizon. 
This problem was recently solved in~\cite{2017-WWW-Drutsa}, where the algorithm PRRFES with a tight regret bound in $\Theta(\log\log T)$ was proposed. Some extensions of this algorithm were proposed in~\cite{2017-ArXiV-Drutsa}.
In contrast to these studies, first, we search for a pricing algorithm that maximizes the  strategic revenue \emph{expected} over buyer valuations, i.e., $\Expect_v[\SRev_{\gb^\ttS,\gb^\ttB}(\A, v)]$,  (equivalently, s.t.\ $\Expect_v[\SReg_{\gb^\ttS,\gb^\ttB}(\A, v)]\rightarrow\min$), which matches the practice of ad exchanges and optimization goals in classical auction theory~\cite{2009-Book-Krishna}.
Second, our revenue optimization problem is solved \emph{exactly} (not approximately and not via optimization of lower/upper bounds).
Third, our study considers a more general setup in which not only the buyer's surplus is discounted over rounds, but also the seller's revenue does.

%

\section{Constant pricing algorithms}
\label{sec_ConstAlg}

We start investigation of the problem from study of \emph{constant algorithms}, i.e., such algorithms that propose only one price over all rounds independently of the buyer's decisions.  
\begin{definition}
	A pricing algorithm $\A$ is said to be \emph{constant}, if there exists a price $ p \in [0; +\infty)$ s.t., at each node $\n \in \N$, the algorithm's price $\A(\n)$  equals $p$. This price $p$ is referred to as \emph{the algorithm price} and  is denoted by $p(\A)$. \emph{The set of all constant algorithms} is denoted by $\Alg_0 \subset \Alg$.
\end{definition}


Note that 
since a constant algorithm $\A \in \Alg_0$ offers a price $p=p(\A)$ that is independent of buyer decisions, the buyer has no incentive to lie and behaves thus truthfully. Hence, the buyer either rejects the price all the rounds, or accepts it (in our notations, applies the strategy $\nO^\infty$ or $\nI^\infty$, resp.) depending on whether his valuation $v$ is lower than $p$ or not. 
Since $\Rev_{\gb^\ttS}(\A, \nO^\infty)=0$ and $\Rev_{\gb^\ttS}(\A, \nI^\infty)=p\sum_{t = 1}^{\infty} \g^\ttS_t$, the expectation of the strategic revenue of the constant algorithm $\A$ is
\begin{equation*}
	\Expect_{V\sim D} \left[ \SRev_{\gb^\ttS,\gb^\ttB}(\A, V) \right] = \Prob [V < p] \cdot \Rev_{\gb^\ttS}(\A, \nO^\infty) + \Prob [V \ge p]  \cdot \Rev_{\gb^\ttS}(\A, \nI^\infty) = \Prob [V \ge p] \cdot p \cdot \Gamma^\ttS.
\end{equation*}
It is easy to see that a constant algorithm $\A$ is optimal if its price $p(\A)$ is the global maximum point of the function $H_D(p) := \Prob [V \ge p] \cdot p$, which is well known in the theory of non-repeated auctions~\cite{1981-MOR-Myerson,2007-Book-Nisan,2009-Book-Krishna}.
The existence of a global maximum point of $H_D(p)$ for our distribution $D$ is shown in Appendix~\ref{app_subsec_globmax_H_D}, and we refer to the leftmost one of them as \emph{the Myerson price} $p^\ast(D)$~\cite{1981-MOR-Myerson}.
Note that this price can be find via the first-order necessary condition $p = (1-F_D(p))/f_D(p)$, when the distribution $D$ has continuous probability density $f_D$ ($F_D$ is its cumulative distribution function).


\begin{definition}
	The constant algorithm $\A \in \Alg_0$ with the price $p(\A)$ equal to the  Myerson price $p^\ast(D)$ of the distribution $D$ is called \emph{the optimal constant algorithm} and is denoted by $\A^\ast_D$. 
\end{definition}

\section{Equal discounts of the seller and the buyer}
\label{sec_EqDiscounts}
In this section, we study the case when the seller and the buyer  discount their utilities equally, i.e., $\gb := \gb^\ttS = \gb^\ttB$, and we use the following notation for the strategic revenue: $\SRev_\gb := \SRev_{\gb, \gb}$.
First of all, we summarize some useful properties of surplus and revenue as functions of the valuation $v$.
\begin{remark}
	\label{remark_strat_prop}
	Let a pricing algorithm $\A\in\Alg$ and the discount sequence $\gb$ be given. For simplicity, we will use the following short notations of surpluses as mappings from the valuation domain:
	$S_{\astr}(v):=\Sur_\gb(\A,v,\astr)$ and $S(v):=\Sur_\gb(\A,v,\astr^{Opt}(\A, v, \gb))$, for which the following  hold:
	\begin{enumerate}
		\item for each strategy $\astr \in \Str$, the surplus $S_{\astr}$ w.r.t.\ this strategy is a linear function of $v$ of the form $S_{\astr}(v)=q_{\astr} v - r_{\astr}$, where $q_{\astr} = \sum_{t=1}^\infty\g_ta_t$ is the \emph{discounted quantity} of purchased goods and $r_{\astr}$ is the discounted revenue of the seller (i.e., $r_{\astr} = \Rev_\gb(\A,\astr)$);
		\item the strategic (optimal) surplus $S$ is convex as a function of $v$, because it is the maximum of a set of linear functions: $S(v)=\max_{\astr\in\Str}S_{\astr}(v)$ (by definition);
		\item the strategic surplus $S(v)$ is non-negative for any $v \ge 0$ since, for the strategy $\astr = \nO^\infty$, we have $ S_\astr(v) = 0$, which implies in turn that $S(v) \ge S_\astr(v) = 0, \fa v\ge 0$;
		\item the derivative $S'(v)$ exists for almost all $v \in[0; +\infty)$ (i.e., it does not exist on a set of  Lebesgue measure zero), because $S(v)$ is convex and is thus absolutely continuous.
		
	\end{enumerate}
\end{remark}


\begin{lemma}\label{Rlemma}
	For any pricing algorithm $\A\in\Alg$, the strategic revenue $R(v):=\SRev_\gb(\A,v)$ is increasing on the valuation domain $[0; +\infty)$, it starts from zero (i.e., $R(0) = 0$), and the random variable $R(V)$ has thus finite non-negative expectation (i.e., $0 \le \Expect\left[R(V)\right] < +\infty$). 
\end{lemma}
\begin{proof}
	We prove only the first claim since the utilized technique will be useful further.
	The other claims are quite simple and are deferred to Appendix~\ref{app_subsec_proof_RLemma} due to space constraints.
	For any two valuations $v_1$ and $v_2\in[0; +\infty)$ s.t. $v_1<v_2$, and two corresponding optimal strategies $\astr^1$ and $\astr^2\in\Str$, i.e., such that  $S(v_j) = S_{\astr^j}(v_j)$, $j=1,2,$ (using the notations from Remark~\ref{remark_strat_prop}),  we have 
	\begin{equation*} \label{increasepair}
		S_{\astr^1}(v_1) \ge S_{\astr^2}(v_1) \quad \hbox{and}
		\quad S_{\astr^2}(v_2) \ge S_{\astr^1}(v_2).
	\end{equation*}
	Therefore, since $S_{\astr^j}, j=1,2,$ are linear, they either coincide (then $r_{\astr^1}=r_{\astr^2}$), or have an intersection point $w$ in $[v_1,v_2]\!\subset\![0; +\infty)$. In the latter case, one gets $S_{\astr^1}(v)\!\ge\!S_{\astr^2}(v) \fa v\in [0,w]$, which implies $-r_{\astr^1}\ge -r_{\astr^2}$ when $v=0$. Hence, we obtain $R(v_2)=r_{\astr^2}\ge r_{\astr^1}=R(v_1)$ for any $v_2>v_1\ge 0$.
\end{proof}

Similarly to the optimal surplus function $S(\cdot)$ and the strategic revenue one $R(\cdot)$, we introduce \emph{the strategic purchased quantity} $Q(\cdot)$ as a map from the valuation domain, i.e., $Q(v) := \sum_{t = 1}^\infty \g_t a^O_t(v)$, where $\{a^O_t(v)\}_{t = 1}^\infty=\astr^{Opt}(\A, v, \gb)$. Note that $S(v)=Q(v)v-R(v)$, for each $v\in[0,+\infty)$.

\begin{lemma}\label{QSlemma}
	Assume that, for a given $v \ge 0$, the derivative $S^\prime(v)$ exists. Then, $Q(v)$ is uniquely defined and equals to $S^\prime(v)$  for any optimal strategy $\astr$ of the buyer that holds the valuation $v$. 
\end{lemma} 
The proof of this lemma is simple and rather technical; it is also deferred to Appendix~\ref{app_subsec_proof_QSlemma} due to space constraints.
Lemma~\ref{QSlemma} together with the identity $\SRev_\gb(\A, v) \!=\! R(v)\!=\!  Q(v) v \!-\! S(v)$ gives us:
\begin{corollary}
	\label{corollary_uniqDefSRev}
	For almost all $v\in[0; +\infty)$, the strategic revenue $\SReg_\gb(\A,v)$  is uniquely defined for any optimal strategy $\astr$ of the buyer that holds the valuation $v$\footnote{Remind that the strategic revenue may not be uniquely defined (see Footnote~\ref{fn_SRev} near the definition of  the strategic revenue).}.
\end{corollary}

\begin{remark} \label{Qremark}
	Function $Q(v)$ is defined almost everywhere and non-decreasing on its domain, since $Q^\prime(v) = S^{\prime \prime}(v)$, which also defined almost everywhere and not less than 0, since $S$ is convex on its domain\footnote{Note that this fact can be proved directly like in Lemma~\ref{Rlemma}.}. Also by the definition $Q(v) \le \Gamma$ and, thus, $Q(+\infty)$ is finite.
\end{remark}


\subsection{Optimality of the constant algorithm with the Myerson price}
\label{subsec_OptAlgProof}
We use notations for the distribution functions: $F(v) := \Prob[V \le v]$ and $G(v) := 1 - F(v) = \Prob[V > v]$.
\begin{lemma}
	\label{lemma_ExpR_to_Int_dQ_dS}
	For the mappings $S(v)$, $R(v)$, and $Q(v)$ the following identity holds:
	\begin{equation*}
		\label{eq_lemma_ExpR_to_Int_dQ_dS}
		\Expect \left[R(V)\right] = \int\limits_{[0; +\infty)} G(v) Q(v) dv + \int\limits_{[0; +\infty)} G(v) v dQ(v)  - \int\limits_{[0; +\infty)} G(v) dS(v).
	\end{equation*}
\end{lemma}
The proof is rather technical, relies on the properties of $S$, $R$, and $Q$ established in the above statements, and is thus deferred to Appendix~\ref{app_subsec_proof_lemma_dQ_dS}.

\begin{theorem} 
	\label{maintheorem}
	Assume the valuation $V \sim D$ and the discount sequence $\gb = \{\g_t\}_{t = 1}^\infty$ satisfy the aforementioned conditions (see Sec.~\ref{sec_Prelim}). Then the expected strategic revenue of an arbitrary pricing algorithm $\A \in \Alg$ is not greater than the one of the optimal constant algorithm $\A^\ast_D$: 
	\begin{equation} \label{maininequality}
		\forall \A \in \Alg \quad \hbox {we have} \quad \Expect \left[ \SRev_\gb(\A, V) \right] \le  \Expect \left[ \SRev_\gb(\A^\ast_D, V) \right].
	\end{equation}
\end{theorem}

\begin{proof}[Proof of Theorem \ref{maintheorem}]
	
	Consider an arbitrary algorithm $\A \in \Alg$ and use the notations $S$, $R$, and $Q$ introduced above. From Lemma~\ref{lemma_ExpR_to_Int_dQ_dS}, we have
	\begin{equation}
		\label{eq_maintheorem_proof_1}
		\Expect \left[R(V)\right] = \!\!\!\!\!\int\limits_{[0; +\infty)}\!\!\!\!\! G(v) Q(v) dv + \!\!\!\!\!\int\limits_{[0; +\infty)}\!\!\!\!\! G(v) v dQ(v) - \!\!\!\int\limits_{[0; +\infty)}\!\!\!\!\! G(v) dS(v) = \!\!\!\!\!\int\limits_{[0; +\infty)}\!\!\!\!\! G(v) v dQ(v),
	\end{equation}
	where the latter identity of Eq.~(\ref{eq_maintheorem_proof_1}) holds due to the facts that $S$ is absolutely continuous on its domain (see Remark \ref{remark_strat_prop}), thus, $\int_{[0; +\infty)}G(v) dS(v) = \int_{[0; +\infty)} G(v) S^\prime(v) dv$, and that $S'(v) = Q(v)$ almost everywhere (see Lemma~\ref{QSlemma}). By definition, we have $H_D(v)=G(v) v$, $\fa v\ge0,$ and, hence, Eq.~(\ref{eq_maintheorem_proof_1}) implies that $\Expect \left[R(V)\right] = \int_{[0; +\infty)} H_D(v) dQ(v)$ can be upper bound by the expression
	\begin{equation}
		\label{eq_maintheorem_proof_2}
		H_D(p^\ast(D)) \cdot \int_{[0; +\infty)} \!\!\!\!\!\!1 dQ(v) = H_D(p^\ast(D)) \cdot (Q(+\infty) - Q(0)) \le H_D(p^\ast(D)) \cdot \Gamma,
	\end{equation}
	where $H_D(v)$ is bounded by its maximum $H_D(p^\ast(D))$, the first identity is due to the fact that $Q$ is non-decreasing on $v$, and non-negative $Q(v)$ is bounded by $\Gamma$ for all $v\ge 0$ (see Remark~\ref{Qremark}).
	Finally, remind that the expected strategic revenue $\Expect\left[\SRev_\gb(\A^\ast_D, V)\right]$ of the optimal constant algorithm $\A^\ast_D$ equals to the right hand side of Eq.~(\ref{eq_maintheorem_proof_2}) (see Sec.~\ref{sec_ConstAlg}).
\end{proof}

Th.~\ref{maintheorem} states that the optimal constant algorithm $\A^\ast_D$ is, in fact, optimal among all pricings~$\Alg$. 

\subsection{Non-uniqueness of the optimal algorithm: ``big deal" pricing} 
\label{subsec_OptAlgNonUnique}

It appears that \emph{the optimal constant algorithm $\A^\ast_D$ is not the unique optimal one}.
We provide an example of applying a general technique for building optimal algorithms of certain form. 
\begin{proposition}
	\label{prop_non_uniqu}
	Let the game have at least 2 rounds (i.e., $\Gamma > \g_1$). If an algorithm $\A_1$ sets the first price $p_1$ equal to $\Gamma p^\ast(D)/ \g_1$ and sets all further prices either $p_t = 0, t\ge2$, if the buyer accepts the first offer, $a_1=1$, or  $p_t = 2 \g_1 p_1/(\Gamma - \g_1), t\ge2$, otherwise; then the algorithm $\A_1$ is optimal. 
\end{proposition}

\begin{proof}
	First, note that the buyer has no incentive to lie after the first round since the algorithm prices $p_t, t\ge3,$ do not depend on his decisions $a_t, t\ge2$. Hence, possible candidates for optimal strategies are $\nO^\infty$, $\nI^\infty$, $\nO\nI^\infty$, and $\nI\nO^\infty$. 
	It easy to see that the optimal buyer strategy in response to $\A_1$ is $\nI^\infty$ for the case $v > p^\ast(D)$ and $\nO^\infty$ for $v < p^\ast(D)$.
	Indeed, if the buyer accepts $p_1$, further offers are for free goods that will be accepted. If the buyer rejects $p_1$, then, for any strategy $\astr \in \Str$ s.t. $a_1=0$, we have 
	\begin{equation}
		\label{eq_prop_non_uniqu_proof_1}
		S_\astr(v) \le (\Gamma -\g_1) ( v - 2 \g_1 p_1/(\Gamma - \g_1) )
		< \Gamma v - 2 \g_1 p_1 < \Gamma v - \g_1 p_1 = S_{\nI^\infty}(v).
	\end{equation}
	Thus, if $S_{\nI^\infty}(v) > 0 = S_{\nO^\infty}(v)$, then $\nI^\infty$ is optimal strategy, and, if $S_{\nI^\infty}(v) < 0$, then Eq.~(\ref{eq_prop_non_uniqu_proof_1}) implies optimality of $\nO^\infty$.
	Finally, note that $S_{\nI^\infty}(v) = \Gamma v - \gamma_1 p_1 = \Gamma (v - p^\ast(D))$ that implies $S_{\nI^\infty}(v) > 0 \Leftrightarrow v > p^\ast(D)$.
	Hence, the expected strategic revenue of $\A_1$ is 
	\begin{equation}
		\label{eq_prop_non_uniqu_proof_2}
		\Expect\left[\SRev_\gb(\A_1, V)\right]
		= \Prob[p^\ast(D) \le V] \cdot \g_1 \Gamma p^\ast(D) /\g_1 
		=  H_D(p^\ast(D)) \Gamma 
		= \Expect\left[\SRev_\gb(\A^\ast_D, V)\right].
		 \vspace{-0.3cm}
	\end{equation}
\end{proof}
	 \vspace{-0.2cm}
The key idea behind the algorithm $\A_1$ is quite simple. Roughly speaking, the seller ``accumulates" all his revenue at the first round by proposing the buyer a ``big deal": to pay a large price at the first round and get all goods in the subsequent rounds for free, or, otherwise, get nothing\footnote{A similar pricing was proposed by~\cite{2013-OR-Kakade} for a class of mechanism environments with multiplicative separability and zero production cost. Their mechanism charges an up-front payment  (before rounds starts) and posts zero price each round obtaining thus truthfulness.  In contrast to that study, the ``big deal" pricing posts a large price at the first round (our setup does not allow an up-front payment) and is not truthful (since the price $p_1 = \Gamma p^\ast(D)/ \g_1$ is accepted by the strategic buyer whose valuation $v > p^\ast(D)$, not $v > p_1$).}.
Note that this optimal pricing algorithm depends both on the discounting $\gb$ and the valuation distribution $D$:
the price $p_1$ is calculated based on the knowledge of the total discounted revenue $\Gamma p^\ast(D) $ that is earned by $\A^\ast_D$ from selling all goods. 
An attentive reader may note that the idea of the aforementioned technique allows, in fact, to build more variants of optimal algorithms by  ``spreading"  the revenue $\Gamma p^\ast(D)$  in a certain way along the rightmost path of the tree $\T$.
In Sections~\ref{sec_NonEqualDisc_BgS} and~\ref{sec_NonEqualDisc_BlS}, we show that \emph{$\A_1$ may remain optimal in the cases when the constant algorithm $\A^\ast_D$ is no longer optimal}.

\section{Less patient seller}
\label{sec_NonEqualDisc_BgS}
Now we are ready to study the cases when the seller and the buyer discounts are different. 
Further, we argue that \emph{the constant algorithm $\A^\ast_D$ is no longer optimal among all algorithms $\Alg$}  in these cases.

We start our investigation from a seller which is less patient than the buyer in willingness to wait for the revenue. We consider the case when $\gb^\ttS \le \gb^\ttB$ (i.e., $\g_t^\ttS \le \g_t^\ttB \fa t\in\mathbb{N}$); e.g., when the discounts decrease geometrically:  $\gb^\ttS = \{\g_\ttS^{t - 1}\}_{t = 1}^\infty$ and $\gb^\ttB = \{\g_\ttB^{t - 1}\}_{t = 1}^\infty$, where $0<\g_\ttS\le\g_\ttB<1$.


\begin{lemma}
	\label{lemma_BgS_upperbound}
	Let $\A\in\Alg$, then the following upper bound for its expected strategic revenue holds:
	\begin{equation}
		\Expect \left[  \SRev_{\gb^\ttS, \gb^\ttB}(\A, V) \right]  \le \Gamma^\ttB \cdot H_D(p^\ast(D))
	\end{equation}	
\end{lemma}
\begin{proof} 
	Let $\astr^{Opt}(\A, v, \gb^\ttB) = \{a^O_t\}_{t = 1}^\infty$, then, using the independence of $\astr^{Opt}$ on the seller's discount, we get
	$\SRev_{\gb^\ttS, \gb^\ttB}(\A, v) = \sum_{t = 1}^{\infty} \g^\ttS_t a^O_t (v - p_t) \le \sum_{t = 1}^{\infty} \g^\ttB_t a^O_t (v - p_t) = \SRev_{\gb^\ttB, \gb^\ttB} (\A, v) = \SRev_{\gb^\ttB}(\A, v).$ Finally, 
	$
	\Expect \left[  \SRev_{\gb^\ttS, \gb^\ttB}(\A, V) \right] \le \Expect \left[\SRev_{\gb^\ttB}(\A, V)\right] \le \Gamma^\ttB \cdot H_D(p^\ast(D)), 
	$ where Theorem~\ref{maintheorem} is applied with $\gb=\gb^\ttB$ to infer the latter inequality.
\end{proof}

\begin{proposition}
	\label{prop_BgS_A1}
	Let $\gb^\ttS$ and $\gb^\ttB$ be the seller and the buyer discounts, respectively, s.t.\ $\gb^\ttS \le \gb^\ttB$. Then the algorithm $\A_1$ from Proposition~\ref{prop_non_uniqu} with $\gb$ set to $\gb^\ttB$ (i.e., with $p_1 = \Gamma^\ttB p^\ast(D)$) is optimal in $\Alg$. 
\end{proposition}
\begin{proof} Since the optimal strategy is independent of the seller's discount, the beginning of the proof is similar to the one of Prop.~\ref{prop_non_uniqu} up to Eq.~(\ref{eq_prop_non_uniqu_proof_2}), where the seller's discount is used for the first time. In our case of different discounts, the identity Eq.~(\ref{eq_prop_non_uniqu_proof_2}) on the expected strategic revenue will have the form
	$
	\Expect\left[\SRev_{\gb^\ttS, \gb^\ttB}(\A_1, V)\right]
	= \Prob[p^\ast(D) \le V] \cdot \g^\ttS_1 \Gamma^\ttB p^\ast(D) /\g^\ttB_1 
	=  H_D(p^\ast(D)) \Gamma^\ttB,
	$
	where we used $\g^\ttS_1=\g^\ttB_1=1$ (see Remark~\ref{remark_first_discount}).
	We see that $\A_1$ achieves the upper bound of Lemma~\ref{lemma_BgS_upperbound} and is thus optimal.
\end{proof}
The relative expected revenue of the optimal algorithm $\A_1$ w.r.t.\ the optimal constant one $\A^\ast_D$ is $\Gamma^\ttB/\Gamma^\ttS$ which is $> 1$, when $\gb^\ttS < \gb^\ttB$; i.e., \emph{the optimal revenue is larger than the one obtained by offering the Myerson price constantly} (in contrast to the equal discount case).
For instance, for geometric discounts $\gb^\ttS = \{\g_\ttS^{t - 1}\}_{t = 1}^\infty$ and $\gb^\ttB = \{\g_\ttB^{t - 1}\}_{t = 1}^\infty$, this revenue improvement ratio $\Gamma^\ttB/\Gamma^\ttS$ is equal to $ (1-\g_\ttS)/(1-\g_\ttB)$ and  goes to $+\infty$ as $\g_\ttB\rightarrow 1-$ for a fixed $\g_\ttS$. 
Moreover, the  algorithm $\A_1$ provides exactly the same expected revenue as if the seller played in the game with the same discount as the buyer one $\gb^\ttB$.
\emph{This result is quite surprising}, because the dominance of the buyer's discount $\gb^\ttB$ over the seller's one $\gb^\ttS$ suggests a hypothesis that the seller should earn lower than with $\gb^\ttB$ (e.g., see the revenue of $\A^\ast_D$). But the ability of the seller to apply the trick of ``accumulation" of all his revenue at the first round (see Sec.~\ref{subsec_OptAlgNonUnique}) allows him to get the payments for all goods discounted by the buyer's $\gb^\ttB$ at the first round and to boost thus his revenue over the constant pricing.

\section{Less patient buyer}
\label{sec_NonEqualDisc_BlS}
In contrast to the previous cases, finding an optimal pricing here is much more difficult problem since the technique used in Sec.~\ref{sec_EqDiscounts} and~\ref{sec_NonEqualDisc_BgS} to upper bound the expected strategic revenue is no longer applicable (because it relies on the condition $\gb^\ttS \le \gb^\ttB$).
As we will see further, in the studied case, the obtained optimal algorithms are not trivial and require derivation of a multivariate analogue of the functional $H_D(\cdot)$ to be found in a multidimensional space. We obtain this functional in Sec.~\ref{subsec_finite_games_theory} and use it to provide extensive analysis of optimal algorithms in Sec.~\ref{subsec_finite_game_studies} and~\ref{subsec_infinite_game_studies}.
%

\begin{definition}
	\label{def_DiscRate}
	For a discount sequence $\gb = \{\g_t\}_{t = 1}^\infty$, we define \emph{the discount rate} sequence $\nub(\gb) := \{\nu_t(\gb)\}_{t = 1}^\infty$ as the sequence of the ratios of consecutive  components of $\gb$: $\nu_t(\gb) := \g_{t + 1}/\g_t$ when $\g_t > 0$, and $\nu_t(\gb) := 0$ when  $\g_t = 0$\footnote{Recall that if $\g_t = 0$ then $\g_{t'}=0$ for any $t' \ge t$, i.e., $\gb$ has no zeros between positive components (see Sec.~\ref{subsec_Setup}). Hence, the discount rate sequence $\nub(\gb)$ has no zeros between positive components as well.}.
\end{definition}

\begin{remark}
	\label{remark_DiscRateInequality}
	Let $\gb^1 = \{\g^1_t\}_{t = 1}^\infty$ and $\gb^2 = \{\g^2_t\}_{t = 1}^\infty$ be some discounts sequences. Then, the condition $\nub(\gb^2) \ge \nub(\gb^1)$ is equivalent to the one that the sequence $\{\g^2_t/\g^1_t\}_{t = 1}^\infty$ is non-decreasing (formally, treating $0/0$ as $+\infty$). The proof of this statement straightforwardly follows from Definition~\ref{def_DiscRate}.
\end{remark}

From here on in this section we consider the discounts $\gb^\ttS $ and $\gb^\ttB$ such that $\nub(\gb^\ttS) \ge \nub(\gb^\ttB)$.
This condition means that the seller is more patient than the buyer \emph{locally at each round} (see Remark~\ref{remark_DiscRateInequality}). In particularly, $\nub(\gb^\ttS) \ge \nub(\gb^\ttB)$ implies that $\gb^\ttS \ge \gb^\ttB$, i.e., the seller is globally more patient than the buyer as well, but the inverse implication is not true\footnote{We believe that the studied case of $\nub(\gb^\ttS) \ge \nub(\gb^\ttB)$ covers a large variety of discount sequences (e.g., the geometric ones) that describe a more patient seller.   Nonetheless, the study of the case when $\gb^\ttS \ge \gb^\ttB$ and $\nub(\gb^\ttS) \not\ge \nub(\gb^\ttB)$ is interesting and is left for future work. A possible direction to study this case consists in our following insight: if the buyer is locally more patient than the seller  at some round $t$ (i.e., $\nu_t(\gb^\ttS) < \nu_t(\gb^\ttB)$), then the trick similar to the one used in the ``big deal" algorithm can be applied at this round $t$ to get an optimal algorithm.}.
A typical example of the studied case is a pair of geometric discounts:  $\gb^\ttS = \{\g_\ttS^{t - 1}\}_{t = 1}^\infty$ and $\gb^\ttB = \{\g_\ttB^{t - 1}\}_{t = 1}^\infty$, where $0<\g_\ttB\le\g_\ttS<1$.


\begin{definition}
	Let $\gb$ be a discount sequence, then an algorithm $\A\in\Alg$ is said to be \emph{completely active for $\gb$}, if for any strategy $\astr\in\Str$ there exists a valuation $v \in [0; +\infty)$ such that $S_\astr(v) = S(v)$, where $S$ and $S_\astr$ are defined in Remark~\ref{remark_strat_prop}, i.e., the surplus function $S_\astr$ is tangent to the optimal surplus function $S$. We denote the set of all completely active algorithms for $\gb$ by $\tilde \Alg(\gb)$. 
\end{definition}

In the next subsection, we will obtain the central results of our study. We do it for the case of a finite number of rounds, but, in Sec.~\ref{subsec_infinite_game_studies}, we show how to use these results to obtain approximately optimal algorithms for the case of the infinite number of rounds.

%


\subsection{Finite games: multivariate optimization functional}
\label{subsec_finite_games_theory}
In this section, we consider the case of the  game with a finite time horizon $T\in\mathbb{N}$: in particular, in this case, seller algorithms, buyer strategies, and all discounts (including $\gb^\ttS, \gb^\ttB$) are considered as their $T$-length variants (they can be defined in a natural way similarly to their infinite analogues). For simplicity of presentation, we assume that all discounts are positive (i.e., $\neq \!\!0$) in all $T$ rounds.


\begin{definition} 
	\label{def_RegularDiscount}
	A discount sequence $\gb$ is said to be \emph{regular}\footnote{The reasons to introduce this  class of  discounts are discussed in Remark~\ref{remark_RegularDiscount}.}, if  $\gb\cdot\astr^1 \neq \gb\cdot\astr^2$ for any pair of strategies $\astr^1, \astr^2\in\Str$, i.e., any buyer strategy $\astr\in\Str$ results in a unique discounted quantity of purchased goods.
	Here we used the short notation for the scalar product: $\astr\cdot\bstr:=\sum_ta_tb_t$.
\end{definition}

In the following important proposition we show that any algorithm can be transformed to a completely active one for the discount $\gb^\ttB$ with no loss in the expected strategic revenue.

\begin{proposition} \label{prop_RatesCompleteActiveness}
	In a $T$-round game, let $\gb^\ttS, \gb^\ttB$ be discounts s.t.\ $\nub(\gb^\ttB) \le \nub(\gb^\ttS)$ and $\gb^\ttB$ is a regular one. Then, for any pricing algorithm $\A \in \Alg$,  there exists a completely active  algorithm $\tilde \A \in \tilde \Alg(\gb^\ttB)$ s.t.\ 
	\begin{equation}
	\label{prop_RatesCompleteActiveness_eq1}
	\Expect\left[\SRev_{\gb^\ttS, \gb^\ttB}(\A, V)\right] \le \Expect\left[\SRev_{\gb^\ttS, \gb^\ttB}(\tilde \A, V)\right].
	\end{equation}
\end{proposition}

\begin{proof}
	For a given algorithm and a given discount $\g^\ttS$, we will use the notation $r_\astr := \Rev_{\gb^\ttS}(\A, \astr)$ for any $\astr\in\Str$ (similarly to Remark~\ref{remark_strat_prop}, but  indicating explicitly the seller's discount).
	The main idea of the proof consists in the following technique.
	We will consider all strategies $\astr$ s.t. $S_\astr(v) < S(v)$  $\fa v \in [0; +\infty)$ (referred to as  \emph{non-active}), and, consequently, for each  of them denoted by $\astr$, we  apply the following procedure of modifying the source algorithm $\A$: define a transformation $\A'$ that does  not change $S_\bstr$ for $\bstr \in\Str\setminus\{\astr\}$, moves $S_\astr$ to the left until it is tangent to $S$ in some $v \in [0; +\infty)$, decreases $r_\astr$, and does not decrease $r_\bstr$ for $\bstr \in\Str\setminus\{\astr\}$. That will imply that the expected strategic revenue of the transformed algorithm $\A'$ is no lower than the one of the source algorithm $\A$. In this way, we will (one-by-one) make all strategies active.

	Let us consider the set of all non-active strategies. If it is empty, then $\A\in \tilde \Alg(\gb^\ttB)$ and Eq.~(\ref{prop_RatesCompleteActiveness_eq1}) holds. 
	Otherwise, note that the ``always-reject" strategy $\astr = \nO^T$ is always active, since $S_\astr(0) = 0 = S(0)$. Hence, one can order all non-active strategies  by ``the last $\nI$ index"   $t_1(\astr) = \max \{t|\ a_t = 1\}$. 
	
	We take a non-active strategy $\astr$ with the smallest $t_1(\astr)$, denoting $t_1 := t_1(\astr)$ and the node \mbox{$\n := a_1 a_2 \dots a_{t_1 - 1}$}, and construct a new algorithm $ \A'$ based on the source one $\A$ in the following way.  
	Set $\A' = \A$ and transform the prices $ \A'(\n),  \A'(\rt(\n)), \ldots,  \A'(\lf^{T - t_1-1}(\rt(\n)))$ as follows: 
	\begin{enumerate}
		\item decrease $\A'(\n)$ until the function $S_\astr$ is tangent to the function $S$ in some $v \in [0; +\infty)$;
		\item if $t_1 < T$, increase $ \A'(\lf^j(\rt(\n)))$ for $j = 0, \dots, T - t_1-1$ in such a way that 
\begin{equation}
\label{prop_RatesCompleteActiveness_proof_eq1}
		\g^\ttB_{t_1} \cdot \A'(\n) + \g^\ttB_{t_1 + j+1} \cdot \A'(\lf^j(\rt(\n))) = \mathrm{const}.
\end{equation}
	\end{enumerate}

	Since we chosen $\astr$ with the smallest $t_1(\astr)$ among non-active strategies the price $\A'(\n)$ obtained in the step~1 is non-negative (and, thus, this step is correct).
	Indeed, substitute the $t_1$-th component in $\astr$ by $0$ and denote the obtained strategy by $\bstr$. Due to selection of $\astr$, the strategy $\bstr$ is active. 
	Therefore, assume $\A'(\n)$ is decreased to $0$, then the function $S_\astr(v)$ becomes equal to $S_\bstr(v) + \g^\ttB_{t_1} v$ by the definition. Since $S_\bstr$ is tangent to $S$, the increase of its slope by $\g^\ttB_{t_1}$ will result in intersection with $S$. This means that $S_\astr$ will be tangent to $S$ before $\A'(\n)$ reaches $0$.
	
	Now let us prove that the transformation $\A'$ satisfies properties announced at the beginning of the proof. 
	Let $\bstr \in\Str\setminus\{\astr\}$. 
	The step~2 implies that the transformation does not change $S_\bstr$.
	For a strategy $\bstr$ that does not come through the node $\rt(\n)$, the revenue $r_\bstr$ remains the same, since the algorithm prices that contribute to $r_\bstr$ are not altered. For $\bstr \neq \astr$ that comes through the node $\rt(\n)$, let us prove that $r_\bstr$ can only increase. Since $\bstr \neq \astr$ there is a round $t=t_1+j+1, j\ge0,$ where  $b_t=1$.  Let $j$ s.t.\ this $t$ is the first round of acceptance after reaching the node $\rt(\n)$, and let us denote the node where this acceptance take place by  $\m:=\lf^j(\rt(\n))$. Therefore, one can write  the following expression for the increment of  $r_\bstr$:
$
\g^\ttS_{t_1} \left( \A'(\n) - \A(\n) + ({\g^\ttS_{t_1 + j+1}}/{\g^\ttS_{t_1}}) \big( \A'(\m) - \A(\m) \big) \right) =  $
$ = \g^\ttS_{t_1} \left(-({\g^\ttB_{t_1 + j+1}}/{\g^\ttB_{t_1}})\big (\A'(\m) - \A(\m)\big) + ({\g^\ttS_{t_1 + j+1}}/{\g^\ttS_{t_1}}) \big(\A'(\m) - \A(\m)\big)\right) 
\ge 0,$
		where we used Eq.~(\ref{prop_RatesCompleteActiveness_proof_eq1}) to obtain the first equation and used $\nub(\gb^\ttB) \le \nub(\gb^\ttS)$ to obtain the last inequality. So, $r_\bstr$ can only increase for $\bstr \in\Str\setminus\{\astr\}$.
	
	Finally, since $S_\astr$ becomes tangent to $S$, which is convex (see Remark~\ref{remark_strat_prop}), the function $S_\astr$ either equals to $S$ exactly in one point $v\in[0; +\infty)$ or coincides with $S_\bstr$ for some $\bstr \in\Str\setminus\{\astr\}$. The latter case is impossible since a function $S_\bstr$ have different slope for different strategy $\bstr$, because of regularity of $\g^\ttB$. Therefore, the optimal strategy does not change for the buyer with any valuation $v$ except the only one s.t.\ $S_\astr(v)=S(v)$, and the strategic revenue expectation is not affected by the decrease of $r_\astr$ (due to continuity of the valuation distribution $D$). 
	Thus, $\Expect\left[\SRev_{\gb^\ttS, \gb^\ttB}(\A, V)\right] \le \Expect\left[\SRev_{\gb^\ttS, \gb^\ttB}( \A', V)\right]$ and the number of non-active strategies of $\A'$ is reduced by one w.r.t.\ $\A$. After that, we repeatedly apply the above described transformation to $\A'$ until the resulted algorithm has no  non-active strategies. In this way, we get $\tilde\A\in\tilde\Alg$ that satisfies Eq.~(\ref{prop_RatesCompleteActiveness_proof_eq1}).
\end{proof}

An attentive reader may note that the the finiteness of the game is crucially used in the assumption that any (non-active) strategy  $\astr$ has ``the last  $\nI$ index" $t_1(\astr)$. It is certainly untrue for infinite strategies since there are the ones that accept the offer infinite number of rounds. Therefore, we consider the validity of the Prop.~\ref{prop_RatesCompleteActiveness}'s statement (or its analogue) for the infinite game as an open research question that could be considered as a possible direction for future work.

\begin{corollary}
	 \label{cor_completely_active_corollary}
	In a $T$-round game, let $\gb^\ttS, \gb^\ttB$ be discounts s.t.\ $\nub(\gb^\ttB) \le \nub(\gb^\ttS)$ and $\gb^\ttB$ is a regular one. If there exists an optimal pricing algorithm $\A^* \in \Alg$,  then there exists an optimal completely active algorithm $\tilde \A^* \in \tilde \Alg(\gb^\ttB)$. Thus, $\max_{\A \in \Alg} \Expect[\SRev_{\gb^\ttS, \gb^\ttB} (\A, V)] = \max_{\tilde\A \in \tilde \Alg(\gb^\ttB)} \Expect[\SRev_{\gb^\ttS, \gb^\ttB} (\tilde\A, V)]$.
\end{corollary}

This corollary can be easily obtained from the previous proposition and tells us that one can search for an optimal pricing algorithm among the class of  completely active ones $\tilde\Alg$.
Our next goal is to show that this class of algorithms $\tilde\Alg$ can be linearly parametrized by the set $\Delta^k:= \{\vb = \{v_j\}_{j=1}^k\in \mathbb{R}^k|\ 0 \le v_1 \le \dots \le v_k \}$, where $k := k(T) := 2^T - 1$. In order to do this, first of all, we introduce several matrix and vector notations.
First, from here on in our paper we fix an order of nodes $\N = \{\n_1,\ldots,\n_k\}$\footnote{E.g., a consistent order: the nodes from the left subtree come before the root node $\estr$, and the ones from the right subtree come after the root $\estr$; then we recursively repeat this rule for the left and right subtrees.}, and, given this, we represent an algorithm $\A\in\Alg$ as the vector of its prices $\A = (\A(\n_1),\ldots,\A(\n_k))$; note we use the same notation both for the algorithm and its vector representation, since the object type could be easily restored from the context where it is used.
We also introduce the map $\pb: \Str\times\Alg \rightarrow \mathbb{R}^T$, where  $\pb(\astr, \A)$ is the vector of consecutively offered prices by the algorithm $\A\in\Alg$ along the path $\astr\in\Str$. 

Second, given a regular discount $\gb$, we introduce the notion of \emph{$\gb$-dependent natural order} of the buyer strategies $\Str = \{\nO,\nI\}^T$: $\astr\prec_\gb\bstr \Leftrightarrow \gb^\ttB \cdot \astr < \gb^\ttB \cdot \bstr$ for any $\astr, \bstr\in\Str$. The important property of this order consists in that the slope of the $\gb$-discounted surplus function $S_\astr$ is lower than the one of $S_\bstr$ when $\astr\prec_\gb\bstr$. Using this order, we index the strategies: $\Str = \{\astr^0,\ldots,\astr^k\}$; note that the strategy $\nO^T$  is always the first one $\astr^0$, while the strategy $\nI^T$ is the last one $\astr^k$.
Third, given another discount $\gb'$, we introduce the payment vector $\rb(\gb', \gb, \A)$, whose $j$-th component is $r_j(\gb', \gb, \A) := \gb' \cdot \pb(\astr^j, \A)$ for $j = 1,\ldots,k$ (note that we exclude the zero payment corresponded to the zeroth strategy $\astr^0$).
We treat all vectors as vector-columns in our matrix operations.

Finally, we introduce the following $k \times k$  matrices: 
	\begin{itemize}
	\item  $J_T$ is a two-diagonal matrix with $1$ on the diagonal and $-1$ under the diagonal;
	\item $Z_T(\gb) = \mathrm{diag}(z_1,\ldots,z_k)$, where  with $z_j = (\gb \cdot \astr^j - \gb \cdot \astr^{j- 1})^{-1}$ for $j = 1,\ldots,k$;
	\item $K_T(\gb, \gb')=((\kappa_{ij}))_{i,j =  1,\ldots,k}$, where $\kappa_{ij} = \g'_ta_t$ if the path $\astr^i\in\Str$ passes through the node $\n_j\in\N$ whose round is $t$\footnote{In other words, the node $\n_j$ can be represented in the string notation as $a^i_1\ldots a^i_{t-1}$ for some $1\le t \le T$ (see Sec.~\ref{subsec_Notations}).}, and $\kappa=0$, otherwise. Note that, by the definition, the $i$-th component of the vector $K_T(\gb, \gb') \A$ is equal to $\sum_{t = 1}^T \g'_t a^i_t \A(a^i_1 \dots a^i_{t - 1})$.
\end{itemize}

\begin{lemma} 
	\label{lemma_linear_transormation}
	In a $T$-round game, let $\gb$ be a regular discount, the strategies $\Str$ are naturally ordered by $\gb$ (as above), while the matrix and vector notations are introduced as above, then the set of completely active pricing algorithms $\tilde \Alg(\gb)$ (i.e., their vector representations) can be linearly mapped onto $\Delta^{k(T)}$ by the matrix $W_T(\gb) := Z_T(\gb)J_TK_T(\gb, \gb)$, which is correctly defined and  is invertible. 
\end{lemma}

\begin{proof}
	First, by the definition of the matrix $K_T(\gb, \gb)$ and the vector $\A$, we have that the payment vector $\rb(\gb, \gb, \A) = K_T(\gb, \gb)\A$.
	Second, let us denote the intersection point of the lines $S_{\astr^j}$ and $S_{\astr^{j - 1}}$ by $v_j$ for $j = 1, \dots, k$ and combine them in the vector $\vb = (v_1,\ldots,v_k)$. 
	From the identities
	$$
	 \gb \cdot \astr^j v_j -  r_j(\gb, \gb, \A) = S_{\astr^j}(v_j) = 	S_{\astr^{j-1}}(v_j)  =\gb \cdot \astr^{j - 1} v_j -  r_j(\gb, \gb, \A)  ,  \qquad j = 1, \dots, k,
	$$
	by simple arithmetic calculations, one can show that these intersection points can be expressed via the payment vector in the following matrix form: $\vb = Z_T(\gb) J_T \rb(\gb, \gb, \A)$. Combining with the previous finding, we have that $\vb = Z_T(\gb) J_TK_T(\gb, \gb)\A$. So, we obtain in this way the linear map $\wb_\gb (\A) := W_T(\gb)\A : \Alg \rightarrow \mathbb{R}^k$  that depends on $\gb$.
	
	The proof of the statement that $\wb_\gb (\A) \in \Delta^{k(T)}$ if and only if $\A \in \tilde \Alg(\gb)$  could be made via two inductions and is rather technical. Hence, it is deferred to Appendix~\ref{app_subsec_proof_lemma_linear_transormation} due to space constraints.
	The matrices $Z_T$, $J_T$, and $K_T$ are invertible\footnote{This fact is trivial for matrices $Z_T$ and $J_T$. To show this for $K_T$, just apply the induction. By rearranging of rows and columns of $K_T$ (it does not affect the property of invertibility)  one can obtain a  block diagonal matrix  with two blocks. Each of these blocks is based on a matrix with the form like $K_{T-1}$.}, thus, both the matrix $W_T$ and the map $\wb_\gb: \Alg \rightarrow \mathbb{R}^k$ are invertible as well. Hence, $\tilde \Alg(\gb)$ is linearly mapped onto $\Delta^{k(T)}$ by $\wb_\gb$.
\end{proof}

\begin{proposition}  
	\label{prop_problem_reduction}
	In a $T$-round game, let $\gb^\ttS$ be a discount, $\gb^\ttB$ be a regular discount, the strategies $\Str$ are naturally ordered by $\gb^\ttB$ (as above), while the matrix and vector notations are introduced as above. 
	Then there exists an invertible linear transformation $\wb_{\gb^\ttB}:\tilde \Alg(\gb^\ttB) \to \Delta^{k}, k = k(T)$ s.t., for any completely active pricing algorithm $\A \in \tilde \Alg(\gb^\ttB)$, its expected strategic revenue has the form
	\begin{equation}
	\label{prop_problem_reduction_eq_1}
\Expect_{V\sim D}\left[\SRev_{\gb^\ttS, \gb^\ttB}(\A, V)\right] = L_{D,\gb^\ttS, \gb^\ttB}(\vb) \quad   \hbox{for} \quad \vb := \wb(\A), 
	\end{equation}
	where
	\begin{equation}
\label{prop_problem_reduction_eq_2}
L_{D,\gb^\ttS, \gb^\ttB}(\vb) := (1 - F_D(\vb))^\intercal  \Xi_T(\gb^\ttS, \gb^\ttB)\vb, \qquad  \vb\in\Delta^{k},
\end{equation}
	 $\Xi_T(\gb^\ttS, \gb^\ttB):= J_T \cdot K_T(\gb^\ttB, \gb^\ttS)  K_T(\gb^\ttB, \gb^\ttB)^{-1}  J_T^{-1}  Z_T(\gb^\ttB)^{-1}$ is the invertible $k\times k$ matrix that depends only on the discounts, the vector $(1 - F_D(\vb)) \in \mathbb{R}^{k}$ has the $i$-th component equal to $1 - F_D(v_i)$, and $F_D$ is the cumulative distribution function of the variable $V$. 
\end{proposition}

\begin{proof}
	Let us take the transformation $\wb_{\gb^\ttB}$ defined by $\wb_{\gb^\ttB} (\A) := W_T(\gb^\ttB)\A$ (as in the proof of Lemma~\ref{lemma_linear_transormation}) and $\vb = \wb_{\gb^\ttB} (\A)$. Recall that, in this case, the $j$-th component of $\vb$ is the intersection point of the straight-line functions $S_{\astr^j}$ and $S_{\astr^{j - 1}}$. It is evident that the strategic buyer chooses the strategy $\astr^j$, when his valuation $v$ is in the segment $[v_{j}; v_{j + 1})$ for $j \ge 0$ (to be formally correct, we set $v_0 := 0, v_{k + 1} := +\infty$). Thus, the expected strategic revenue equals to
	$$
	\Expect\left[\SRev_{\gb^\ttS, \gb^\ttB}(\A, V)\right] = \sum_{j = 1}^k (F_D(v_{j + 1}) - F_D(v_j))  (\gb^\ttS \cdot \pb(\astr^j, \A)) = \sum_{j = 1}^k (F_D(v_{j + 1}) - F_D(v_j))  r_j(\gb^\ttS,\gb^\ttB , \A),
	$$
	see the definitions of $\pb$ and $\rb$ before Lemma~\ref{lemma_linear_transormation}.
	Let us denote by $dF(\vb)$ the $k$-dimensional vector with  $F_D(v_{j + 1}) - F_D(v_j)$ in the $j$-th component, then, using the identity $dF(\vb) = J_T^\intercal (1 - F_D(\vb))$, we have
	$$
		\Expect\left[\SRev_{\gb^\ttS, \gb^\ttB}(\A, V)\right] = dF(v)^\intercal  \rb(\gb^\ttS,\gb^\ttB , \A) = (1 - F_D(v))^\intercal  J_T \rb(\gb^\ttS,\gb^\ttB , \A).
	$$
	From the definition of the matrix $K_T$, one can obtain  $\rb(\gb^\ttS,\gb^\ttB , \A) = K_T(\gb^\ttB, \gb^\ttS) \A$ (as in the proof of Lemma~\ref{lemma_linear_transormation}). 
	Finally, we have $\A = W_T(\gb^\ttB)^{-1} \vb  = K_T(\gb^\ttB, \gb^\ttB)^{-1}  J_T^{-1} Z_T(\gb^\ttB)^{-1} \vb$ due to $\vb = W_T(\gb^\ttB)\A$ and  invertibility of $\wb_{\g^{\ttB}}$.
	
	So, let us combine all together:
	$$
		 \Expect\left[\SRev_{\gb^\ttS, \gb^\ttB}(\A, V)\right] = (1 - F_D(\vb))^\intercal  J_T \cdot K_T(\gb^\ttB, \gb^\ttS)  K_T(\gb^\ttB, \gb^\ttB)^{-1}  J_T^{-1}  Z_T(\gb^\ttB)^{-1}  \vb,
	$$
	where the matrix product between $(1 - F_D(\vb))^\intercal$ and $\vb$ is exactly the matrix $\Xi_T(\gb^\ttS, \gb^\ttB)$. 
\end{proof}

Corollary~\ref{cor_completely_active_corollary} and Proposition~\ref{prop_problem_reduction} immediately infer the following key result of our study.
\begin{theorem} 
	\label{th_problem_equivalence}
	In a $T$-round game, let $\gb^\ttS, \gb^\ttB$ be discounts s.t.\ $\nub(\gb^\ttB) \le \nub(\gb^\ttS)$ and $\gb^\ttB$ is a regular one.
	The optimization problem of finding an optimal algorithm is equivalent to maximization of the multivariate functional $L_{D,\gb^\ttS, \gb^\ttB}(\cdot)$ over the set $\Delta^k = \{\vb \in \mathbb{R}^k|\ 0 \le v_1 \le \dots \le v_k \}$,  $k = 2^T-1$, i.e.,
	\begin{equation}
	\label{th_problem_equivalence_eq1}
	\max_{\A \in \Alg} \Expect_{V\sim D}\left[\SRev_{\gb^\ttB, \gb^\ttS} (\A, V)\right] = \max_{\vb \in \Delta^k} L_{D,\gb^\ttS, \gb^\ttB}(\vb),
	\end{equation}
	 where $L_{D,\gb^\ttS, \gb^\ttB}$ is defined in Eq.~(\ref{prop_problem_reduction_eq_2}) and depends only on the discounts and the distribution $D$ of the valuation variable $V$.
\end{theorem}

It is quite important to emphasize that the $k$-dimensional functional $L_{D,\gb^\ttS, \gb^\ttB}$  is a \emph{bilinear form} applied to the vectors $\vb$ and $1-F_D(\vb)$. 
This bilinear form is independent of the distribution $D$ and is defined by the matrix $\Xi_T(\gb^\ttS, \gb^\ttB)$.
 In this view, we note that there is a strong relationship between our optimization functional $L_{D,\gb^\ttS, \gb^\ttB}$ and the function  $H_D$ (see Sec.~\ref{sec_ConstAlg}).
In other words, the functional $L_{D,\gb^\ttS, \gb^\ttB}$ constitutes the key basis of optimal algorithms and is fundamental for them as the function $H_D(p) = p\Prob_{V\sim D}[V\ge p ]$ is fundamental for  optimal pricing in static auctions.

\begin{figure}
	\centering
	\vspace{-2mm}
	\includegraphics[width=\columnwidth]{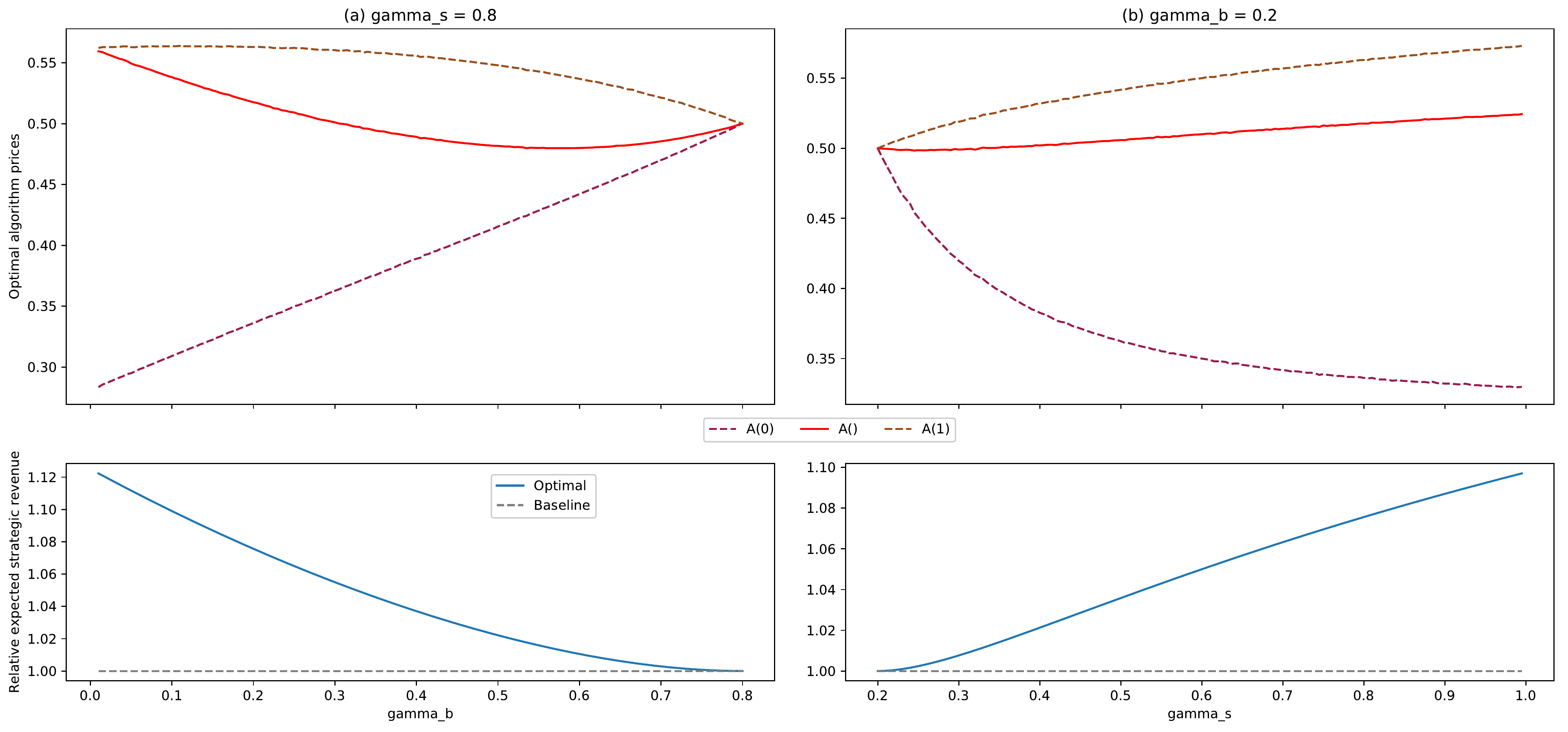}
	\vspace{-8mm}
	\caption{$2$-round game. The prices  $\A^\ast(\nO), \A^\ast(\estr), \A^\ast(\nI)$ and the relative expected strategic revenue (w.r.t.\ $\A^*_D$) of the optimal algorithm $\A^\ast$ for discount rates:
		(a) $\g_\ttS = 0.8$ and various $\g_\ttB$;
		(b) $\g_\ttB = 0.2$ and various $\g_\ttS$.}
	\label{img_T2_Uniform}
	\vspace{-2mm}
\end{figure}

\begin{remark}[Th.~\ref{maintheorem} as a special case of Th.~\ref{th_problem_equivalence}] 
	\label{remark_Lfunctional_forequal}
	Let us consider the case of equal discounts, $\gb^\ttS = \gb^\ttB$, then $K_T(\gb^\ttB, \gb^\ttS)= K_T(\gb^\ttB, \gb^\ttB)$ and the matrix $\Xi_T(\gb^\ttS, \gb^\ttB) = J_T \cdot K_T(\gb^\ttB, \gb^\ttS)  K_T(\gb^\ttB, \gb^\ttB)^{-1}  J_T^{-1}  Z_T(\gb^\ttB)^{-1}$ becomes equal just to the diagonal matrix $ Z_T(\gb^\ttB)^{-1}=\mathrm{diag}(\alpha_1,\ldots,\alpha_k)$, $\alpha_j \!\!=\!\! \gb^\ttB \cdot \astr^j - \gb^\ttB \cdot \astr^{j- 1}$. Hence,
	$$\textstyle L_{D,\gb^\ttB, \gb^\ttB} (\vb) = (1 - F_D(\vb))^\intercal   Z_T(\gb^\ttB)^{-1}\vb  = \sum_{j=1}^k(1-F_D(v_j))\alpha_jv_j  =  \sum_{j=1}^kH_D(v_{j})\alpha_j.$$
	Since $\alpha_j>0$ (due to the dependence of the order of $\{\astr^j\}_j$ on $ \gb^\ttB$) and $H_D(v) \le H_D(p^*(D)), \fa v,$ (see Sec.~\ref{sec_ConstAlg}) we infer that this sum above is maximal when $v_1 = \ldots = v_k = p^*(D)$.
    Thus, in the case of equal discounts, the optimization of the functional $L_{D,\gb^\ttB, \gb^\ttB}$ reduces to the maximization of the function $H_D$ used to find  Myerson's price $p^*(D)$. This is expected and \emph{additionally highlights the strong  similarity of our optimization functional for the dynamic pricing to the one for the static pricing}.
\end{remark}

So, in the particular case of equal discounts, the optimization of $L_{D,\gb^\ttB, \gb^\ttB}$ has no closed form solution since it reduces to the optimization of $H_D$. Hence, we expect that, in the other cases, generally, our optimization problem does not admit a closed form solution as well.
In the next subsections, we numerically find the maximum of $L_{D,\gb^\ttS, \gb^\ttB}$ for several representative games and show that the obtained optimal algorithms are no longer constant and significantly outperform the optimal constant pricing in terms of the expected strategic revenue.

\begin{remark}[on regularity of $\gb^\ttB$]
	\label{remark_RegularDiscount}
	The regularity of the discount $\gb^\ttB$ is used  in two cases, namely, to get: 
	(1)~the uniqueness of $\gb$-dependent natural order of the strategies $\Str$;
	(2)~zero probability of the set of the valuations for which the  optimal buyer strategy is not unique.
	The case~(1) is used in Lemma~\ref{lemma_linear_transormation} and Prop.~\ref{prop_problem_reduction}; there, regularity is just needed for simplicity of presentation of the proofs; these statements (possibly with a slight change) will certainly hold without this restriction on $\gb^\ttB$.
	The case~(2) is used in Prop.~\ref{prop_RatesCompleteActiveness} to guarantee that the strategic buyer will not prefer (with non-zero probability) a strategy that has been non-active before the transformation. So, Prop.~\ref{prop_RatesCompleteActiveness} may not hold without regularity of  $\gb^\ttB$. But we believe that one can obtain a similar result for a series of algorithms that "converges" to a one from $\tilde \Alg$  and use this series to obtain the statement of Th.~\ref{th_problem_equivalence}.
 	In any way, the restriction on the regularity of $\gb^\ttB$ does not harm the main conclusions of our work, because, for a finite horizon, regular discounts are more frequent than non-regular ones, e.g., there is just a finite number of non-regular  geometric discounts for a finite horizon. Hence, our qualitative results from Sec.~\ref{subsec_finite_game_studies} and~\ref{subsec_infinite_game_studies} are not affected by this restriction.
\end{remark}

\begin{figure}
	\centering
	\vspace{-2mm}
	\includegraphics[width=\columnwidth]{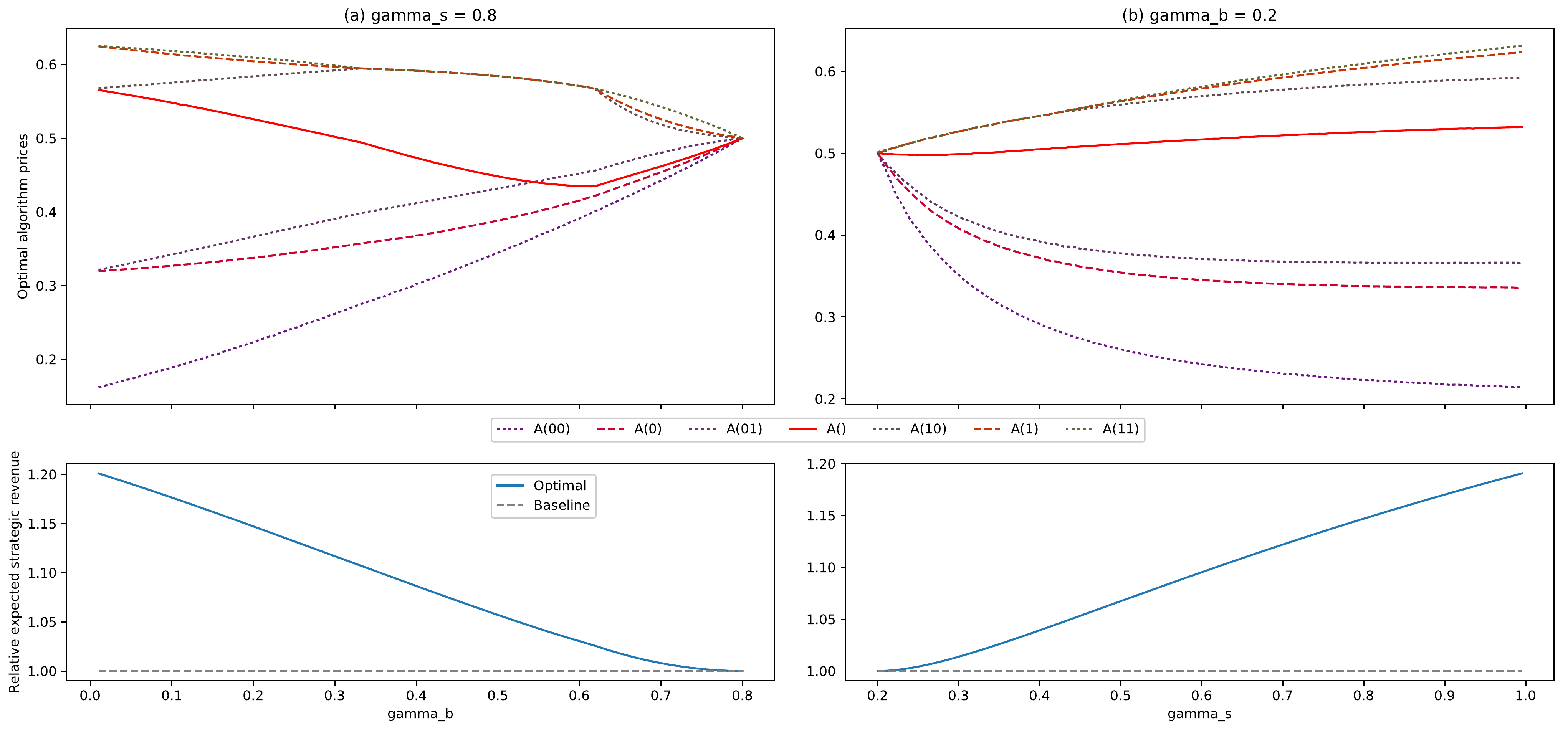}
	\vspace{-8mm}
	\caption{$3$-round game. The prices  $\A^\ast(\n)$, for nodes $\n\in\N$ s.t.\  $|\n|\le 2$, and relative expected strategic revenue (w.r.t.\ $\A^*_D$) of the optimal algorithm $\A^\ast$ for discounts:
		(a) $\g_\ttS = 0.8$ and various $\g_\ttB$;
		(b)  $\g_\ttB = 0.2$ and various $\g_\ttS$.}
	\label{img_T3_Uniform}
	\vspace{-4mm}
\end{figure}

\subsection{Finite games: case study} 
\label{subsec_finite_game_studies}

In this subsection, based on several representative game settings, we demonstrate how to find optimal algorithms using the functional $L_{D,\gb^\ttS, \gb^\ttB}$ and show the key properties of these algorithms. 
We consider finite geometric discounts $\gb^\ttB = \{\g_\ttB^{t - 1}  \Ind_{\{t \le T\}}\}_{t = 1}^\infty, \gb^\ttS = \{\g_\ttS^{t - 1}  \Ind_{\{t \le T\}}\}_{t = 1}^\infty$ for $0 < \g_\ttB < \g_\ttS < 1$ and 
the valuation $V$ uniformly\footnote{During our experimentation, we also analyzed some other distributions. Since the results for them are found to be similar to the ones for uniform, we present these results in Appendix~\ref{app_sec_numsulutions_diff_distr}.} distributed in $[0,1]$. For $2$- and $3$-round games, according to Th.~\ref{th_problem_equivalence}, we find the optimal pricing algorithms by maximizing the functional $L_{D,\gb^\ttS, \gb^\ttB}$ from Eq.~(\ref{prop_problem_reduction_eq_2}).
Their expected revenues are compared with the expected revenue $H_D(p^*(D))\Gamma^\ttS$ of the optimal constant pricing $\A^*_D$ (see Sec.~\ref{sec_ConstAlg}), which is \emph{treated as the baseline} from here on in this paper.

{\bf The case of $T=2$.}
The maximum of the $3$-variate functional $L_{D,\gb^\ttS, \gb^\ttB}$ can be found in the hyperplane $v_2 = v_3$ (the proof is provided in  Appendix~\ref{app_subsec_dimension_reduction}). 
Thus, for $T=2$  the maximization problem is reduced\footnote{This case show that even though the dimension of the problem in right-hand side of Eq.~(\ref{th_problem_equivalence_eq1}) can be reduced, it still could not be reduced to a one-dimensional problem in general. The same we observe in the case of $T=3$.} to a $2$-variate optimization of the function $L_2:\Delta^2 \to \mathbb{R}$, where
$L_2(v_1,v_2) = (1 - F_D(\vb))^\intercal  \Upsilon_2(\g_\ttS, \g_\ttB)\vb$,  $\vb\in\Delta^2$, and $\Upsilon_2(\g_\ttS, \g_\ttB) = 
\begin{pmatrix}
\gamma_\ttS & 0 \\ 
- (\gamma_\ttS - \gamma_\ttB) & 1 + \gamma_\ttS - \gamma_\ttB
\end{pmatrix} $.

Note that, for the uniform distribution $D = U[0; 1]$, $F_D(v) = v$ and the optimized functional $L_2$ becomes thus quadratic. Hence the problem can be solved by means of QP. We solve this problem numerically using the Sequential Least Squares Programming method. 
So, for several pairs of $(\g_\ttS, \g_\ttB)$, we find the optimal algorithm $\A^*$ and depict in Fig.~\ref{img_T2_Uniform} both its prices $\A^*(\n)$ for all nodes $\n$ and its relative expected strategic revenue (w.r.t.\ $\A^*_D$). Namely, Fig.~\ref{img_T2_Uniform}(a) contains results for $\g_\ttS = 0.8$ and $\g_\ttB\in\{0.01 + i \cdot 0.005\}_{i = 0}^{148}$, while Fig.~\ref{img_T2_Uniform}(b) contains results for $\g_\ttB = 0.2$ and $\g_\ttS\in\{0.2 + i \cdot 0.005\}_{i = 0}^{159}$.


First, at the bottom of Fig.~\ref{img_T2_Uniform}  we see that \emph{the optimal algorithm outperforms the baseline optimal constant pricing for any observed pair of discounts}. 
Second, the top part of Fig.~\ref{img_T2_Uniform} demonstrates us that, for any pair of discounts,  the optimal algorithm is  a \emph{consistent pricing}, i.e., the one which never sets prices lower (higher) than earlier accepted (rejected, resp.) ones~\cite{2017-WWW-Drutsa}. In fact, this property is theoretically guaranteed for the studied case; namely, it easily follows from the relation between the optimal prices and the optimum $\vb$: $\A^*(\nO) = v_1$, $\A^*(\estr) = \g_\ttB v_1 + (1 - \g_\ttB) v_2 $, and $\A^*(\nI) = v_3$.
Third, the obtained optimal algorithms are appeared to be continuous in $\g_\ttS$ and $\g_\ttB$. Moreover, if the distance between the discount rates $\g_\ttS$ and $\g_\ttB$ converges to $0$, then the optimal algorithm $\A^*$ converges to the optimal constant one $\A^*_D$ (what experimentally supports Remark~\ref{remark_Lfunctional_forequal}). 

{\bf The case of $T=3$.}
In a similar way as it done for the previous case, the dimensionality of the optimization problem can be lowered from $7$ to $4$, when $\g_\ttB < (\sqrt 5 - 1)/2$, and to $5$, when $\g_\ttB > (\sqrt 5 - 1)/2$\footnote{The different cases are results of the change of the order of the values $\{\gb^\ttB \cdot \astr|\astr\in\Str\}$ at the border point $(\sqrt 5 - 1)/2$.}.
The method to solve the optimization problem and the  set of $(\g_\ttS, \g_\ttB)$ are the same as in the case of $T=2$. Fig.~\ref{img_T3_Uniform} is arranged similarly to Fig.~\ref{img_T2_Uniform}.

Analogously to the case of $T=2$, in Fig.~\ref{img_T3_Uniform}, we observe the superiority of the optimal algorithm $\A^*$ over the baseline $\A^*_D$ for any pair of discount rates, as well as convergence to $\A^*_D$ as $|\g_\ttS - \g_\ttB|\to 0$  and the  continuity of  $\A^*$ in $\g_\ttS$ and $\g_\ttB$.
But, in contrast to the the case of $T=2$, \emph{the optimal algorithm may be non-consistent}: the condition of consistency is violated by the reverse order of the prices $\A^*(\estr) < \A^*(\nO\nI)$ for $\g_\ttB>$  $\approx 0.54$ (which seen in Fig.~\ref{img_T3_Uniform}(a)), i.e., the seller offers a price larger than the one at the first round if the buyer rejects the first price, but accepts the one at the second round.

There is a lot of other interesting observations: e.g., pairs of equal prices  when $\g_\ttB \to 0$ (see Fig.~\ref{img_T2_Uniform} and~\ref{img_T3_Uniform}); some specific area of pairs of   $(\g_\ttS, \g_\ttB)$ where algorithm prices becomes equal (see Fig.~\ref{img_T3_Uniform}), etc. They are seen further in  Fig.~\ref{img_T4_UniformInf} as well, and  a thorough  study of them is deferred to  future work.

\begin{figure}
	\centering
	\vspace{-2mm}
	\includegraphics[width=\columnwidth]{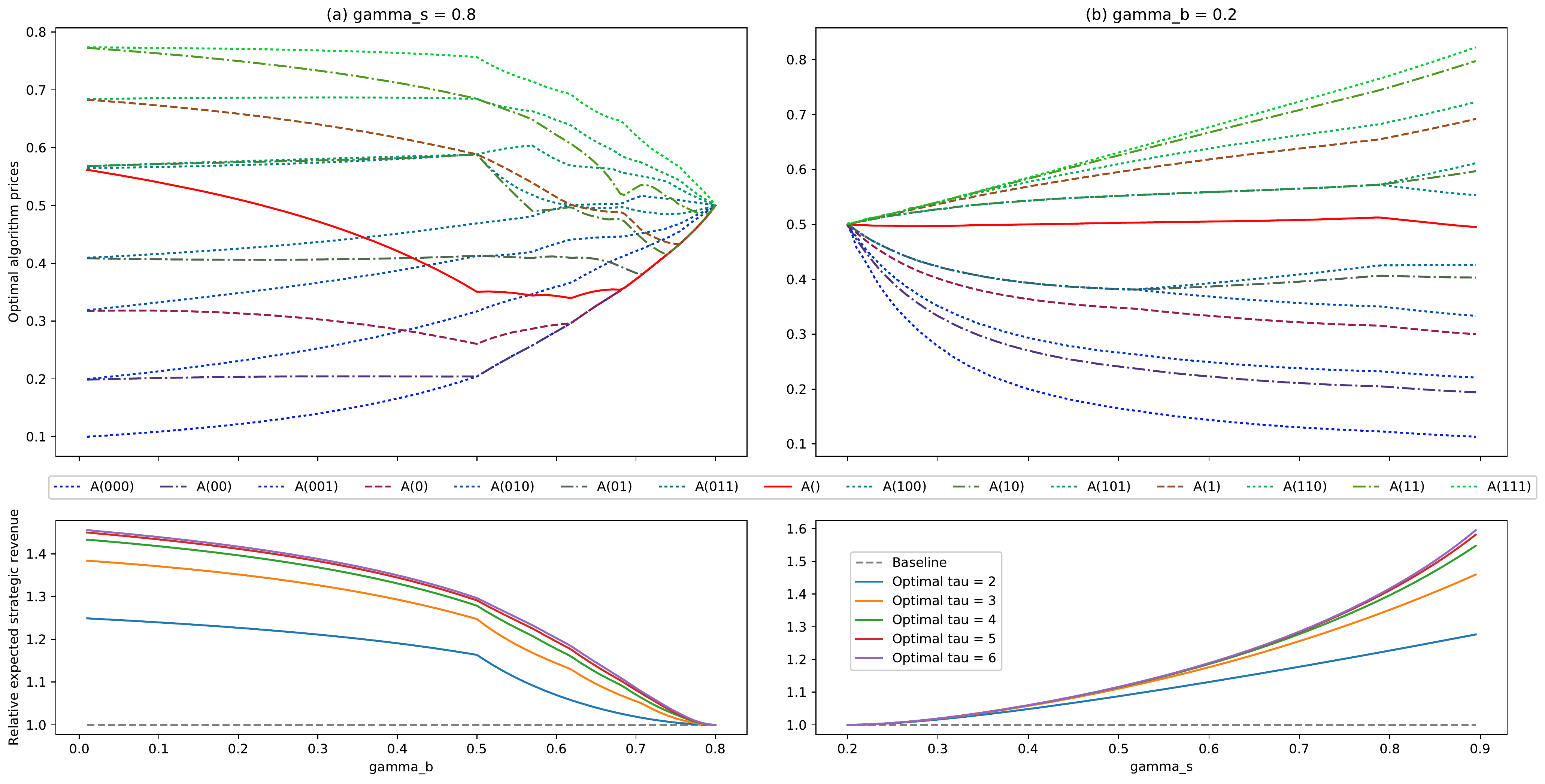}
	\vspace{-8mm}
	\caption{Infinite game. The prices  $\A^\ast_4(\n)$, for nodes $\n\in\N$ s.t.\  $|\n|\le 3$, of the optimal $4$-step algorithm $\A^\ast_4$ and the relative expected strategic revenue (w.r.t.\ $\A^*_D$) of the optimal $\tau$-step algorithm $\A^\ast_\tau, \tau=2,..,6,$ for discounts:
		(a) $\g_\ttS = 0.8$ and various $\g_\ttB$;
		(b)  $\g_\ttB = 0.2$ and various $\g_\ttS$.}
	\label{img_T4_UniformInf}
	\vspace{-4mm}
\end{figure}

\subsection{Infinite game: approximately optimal algorithms and case study} 
\label{subsec_infinite_game_studies}
Let us return to the case of the infinite game with $\nub(\gb^\ttB) < \nub(\gb^\ttS)$. In this case, we have no powerful instrument to find an optimal pricing (unlike to the case of finite games).
However, one can approximate the optimal algorithm by an optimal one in some finite dimensional subclass of $\Alg$. 
Namely, for $\tau\in\mathbb{N}$, let us say that $\A$ is a \emph{$\tau$-step pricing algorithm}, if   \mbox{$\fa \astr, t > \tau:\ \A(\astr_{1: t - 1}) = \A(\astr_{1: \tau - 1}), $} i.e., at rounds $t>\tau$, it offers the price equal to the one that has been offered at the round $\tau$. 
The set of all $\tau$-step algorithms is denoted by $\Alg_{\tau}$ and we refer to any $\A\in\Alg_{\tau}$ as a \emph{finite algorithm} as well.

An attentive reader may note that the problem of finding  the optimal $\tau$-step algorithm for the infinite game is equivalent to finding the optimal algorithm in the  $\tau$-round game with the finite discounts $\tilde \gb^\ttS$ and $\tilde \gb^\ttB$, 
where the operator $\tilde\gb$ means: $\tilde\g_t := \g_t, t<\tau$; 
$\tilde\g_\tau :=\sum_{t = \tau}^\infty \g_t$; and $\tilde\g_t := 0, t > \tau$. The condition $\nub(\gb^\ttB) < \nub(\gb^\ttS)$ implies $\nub(\tilde\gb^\ttB) < \nub(\tilde\gb^\ttS)$ (see Remark~\ref{remark_DiscRateInequality}). Hence, one can apply the optimization technique from Theorem~\ref{th_problem_equivalence}.
The following proposition (the proof  is presented in Appendix~\ref{app_subsec_proof_prop_approx_by_finite_alg}) formally states that the expected strategic revenue of the optimal $\tau$-step algorithm $\A^*_\tau$ converges to one of the optimal pricing $\A^*$ when $\tau\to\infty$.

\begin{proposition}
	\label{prop_approx_by_finite_alg}
	let $\gb^\ttS$, $\gb^\ttB$ be discounts s.t.\ $\nub(\gb^\ttB) \le \nub(\gb^\ttS)$ and $\Gamma^\ttS_\tau := \sum_{t = \tau + 1}^\infty \g^\ttS_t$ for $\tau\in\mathbb{N}$. Then the following bounds hold:
	$$
	\max_{\A \in \Alg_\tau} \Expect\left[\SRev_{\gb^\ttS, \gb^\ttB}(\A, V)\right]  \le  \max_{\A \in \Alg} \Expect\left[\SRev_{\gb^\ttS, \gb^\ttB}(\A, V)\right] \le \max_{\A \in \Alg_\tau} \Expect\left[\SRev_{\gb^\ttS, \gb^\ttB}(\A, V)\right]  + \Gamma^\ttS_\tau \Expect\left[V\right].
	$$
\end{proposition}

Finally, let us consider geometric discounts $\gb^\ttB = \{\g_\ttB^{t - 1} \}_{t = 1}^\infty, \gb^\ttS = \{\g_\ttS^{t - 1} \}_{t = 1}^\infty$ for $0 < \g_\ttB < \g_\ttS < 1$ and 
the valuation $V$ uniformly distributed in $[0,1]$. Following the procedure described in Sec.~\ref{subsec_finite_game_studies}, we numerically find optimal $\tau$-step algorithm $\A^*_\tau$, $\tau=2,..,6$, for the same set of pairs $(\g_\ttS, \g_\ttB)$ as in the case of $T=2$ in Sec.~\ref{subsec_finite_game_studies}. The obtained in this way prices of $\A^*_4$ and the relative expected revenue of $\A^*_\tau$, $\tau=2,..,6$  are arranged in  Fig.~\ref{img_T4_UniformInf} similarly to Fig.~\ref{img_T2_Uniform}.
We see that the expected strategic revenue of $\A^*_\tau$ converges quite quickly to the optimal one. This observation constitutes the empirical evidence of Prop.~\ref{prop_approx_by_finite_alg}, which suggests that the convergence rate is equal to  $\g_\ttS$. We also can note the observations similar to the ones made for the optimal algorithms in Sec.~\ref{subsec_finite_game_studies}. In particular, we see that in the case of the infinite game,  \emph{the baseline optimal constant algorithm is significantly outperformed by algorithms with noticeably non-static pricing} as well.

\section{Conclusions}

We studied online learning algorithms that maximize expected cumulative revenue of repeated posted-price auctions  in the scenario with a strategic buyer that holds a fixed private valuation.
More precisely, we investigated the situation in which the seller ant the buyer may have different level of the patience to wait for utility, and which is modeled via own discounts in the cumulative utilities for the buyer and the seller.
Surprisingly, we found that only in the case of equal discounts, the seller cannot advantageously use the ability to change prices in dynamic fashion (i.e., to learn them) with respect to  the static approach. 
Namely, the case of  equal discounts admits  two optimal algorithms; one of them constantly offers the Myerson price, while the other one  proposes a ``big deal": pay for all goods in advance (at the first round) or get nothing.
But, first, in the case of more patient buyer, the  pricing algorithm ``big deal" was shown to outperform the constant pricing.

Second, in the inverse case when the seller's discount rate is larger than the one of the buyer, 
we reduced the problem of finding an optimal algorithm to a multidimensional optimization problem with a multivariate analogue of the functional used to determine Myerson's price. Our reduction does not admit a closed form solution in general (similarly, to the case of revenue optimal static auctions), but  can be solved by state-of-the-art numerical optimization techniques (like gradient ones).
We conducted extensive analysis of numerically found optimal algorithms to demonstrate that they are non-trivial, may be non-consistent, and generate  larger expected revenue than the constant pricing with the Myerson price.
Overall, this work provided clear techniques for obtaining guarantees on the seller's revenue in repeated posted-price auctions that may help in future studies on a more sophisticated scenarios and auction mechanisms.

%
%

\bibliographystyle{abbrv}
\bibliography{2018-arxiv-rppa}

\appendix
\numberwithin{equation}{section}
\numberwithin{figure}{section}

\section{Missing proofs}

\subsection{Existence of an optimal strategy $\astr \in \Str$ (from Section~\ref{subsec_Setup})}
\label{app_subsec_optstrat_exist}
Assume we are given an algorithm $\A \in \Alg$, a correct discount sequence $\gb = \{\g_t\}_{t = 1}^\infty$ and a private valuation $v \in [0; +\infty)$ (they  are fixed). In this case, for the function $F:\Str \rightarrow \mathbb{R} \cup \{-\infty\},\ F(\astr) = S_\astr(v)$ the following proposition holds. 
\begin{proposition}
	\label{appx_prop_existOptStrat}
	There exists a strategy $\astr^\ast \in \Str$ such that $\forall \astr \in \Str: F(\astr^\ast) \ge F(\astr)$.
\end{proposition}

\begin{proof}[Proof of Proposition~\ref{appx_prop_existOptStrat}]
	Denote $M = \sup_{\astr \in \Str} S_\astr(v)$ and $\Str_0 = \Str$. Let $k \ge 0$ be a non-negative integer, then assume that $a^\ast_1, \dots, a^\ast_k \in \{\nO, \nI\}$ and $\Str_k = \{\astr = \{a_t\}_{t = 1}^\infty \in \Str|\ a_1\dots a_k = a^\ast_1\dots a^\ast_k\}$ such that $\sup_{\astr \in \Str_k} S_\astr(v) = M$ defined (if such conditions holds we call the tuple $(a^\ast_1, \dots, a^\ast_k, \Str_k)$ \emph{correct}). We define such $a^\ast_{k + 1}$ that the tuple $(a^\ast_1, \dots, a^\ast_k, a^\ast_{k + 1}, \Str_{k + 1})$ is correct. 
	
	$\{\astr = \{a_t\}_{t = 1}^\infty \in \Str|\ a_1\dots a_k a_{k + 1} = a^\ast_1\dots a^\ast_k \nO\}$ and $\{\astr = \{a_t\}_{t = 1}^\infty \in \Str|\ a_1\dots a_k a_{k + 1} = a^\ast_1\dots a^\ast_k \nI\}$ are denoted by $\Str^0_k$ and $\Str^1_k$ respectively, similarly $\sup_{\astr \in \Str^0_k} S_\astr(v)$ and $\sup_{\astr \in \Str^1_k} S_\astr(v)$ are denoted by $M^0$ and $M^1$. Since $\Str_k = \Str^0_k \cup \Str^1_k$, the following identity holds $M = \max\{M^0, M^1\}$. If $M = M^0$ we define $a^\ast_{k + 1} = \nO$ and $\Str_{k + 1} = \Str^0_k$, otherwise $a^\ast_{k + 1} = \nI$ and $\Str_{k + 1} = \Str^1_k$. Thus, by the definition of $a^\ast_{k + 1}$ and $\Str_{k + 1}$ the tuple $(a^\ast_1, \dots, a^\ast_k, a^\ast_{k + 1}, \Str_{k + 1})$ is correct. 
	
	Taking into account that for $k = 0$ the tuple $(\Str_0)$ is correct, we obtain uniquely defined sequence $\astr^\ast = \{a^\ast_t\}_{t = 1}^\infty$ such that for all integer $k \ge 0$ the tuple $(a^\ast_1, \dots, a^\ast_k, \Str_k)$ for $\Str_k = \{\astr = \{a_t\}_{t = 1}^\infty \in \Str|\ a_1\dots a_k = a^\ast_1\dots a^\ast_k\}$ is correct.
	
	Before finally proving Proposition~1 we prove an auxillary statement: 
	\[ 
	\forall T \in \mathbb{N}\ \sum_{t \ge T} \g_t a^\ast_t (v - p_t) \ge 0,
	\] where $\{p_t\}_{t = 1}^\infty$ is the sequence of prices set by the algorithm $\A$ in response to $\astr^\ast$. Indeed, for an arbitrary $T \in \mathbb{N}$ assume $\sum_{t \ge T} \g_t a^\ast_t(v - p_t) = - \delta < 0$. In this case, since the series $\sum_{t \ge T} \g_t$ and $\sum_{t \ge T} \g_t a^\ast_t(v - p_t)$ converge, there exists such integer $T_0 \ge T$ that $\sum_{t \ge T_0} \g_t v< \frac \delta 3$  and $\sum_{t \ge T_0} \g_t a^\ast_t(v - p_t) > - \frac \delta 2$, which impllies 
	\[ 
	\sum_{T_0 \ge t \ge T} \g_t a^\ast_t(v - p_t) = \sum_{t \ge T} \g_t a^\ast_t(v - p_t) - \sum_{t > T_0} \g_t a^\ast_t (v - p_t) < -\delta + \frac \delta 2 = - \frac \delta 2
	\]Thus, for an arbitrary strategy $\astr \in \Str_{T_0}$ (denote prices corresponding to $\astr$ by $q_t: \forall t \le T_0 + 1:\ q_t = p_t$) we gain for $\bstr = a^\ast_1\dots a^\ast_T \nO^\infty$
	\[
	S_\astr(v) = \sum_{t \le T} \g_t a^\ast_t (v - p_t) + \sum_{T < t \le T_0} \g_t a^\ast_t(v - p_t) + \sum_{T_0 < t} \g_t a_t (v - q_t) < S_\bstr(v)  - \frac \delta 2 + \frac \delta 3 < M - \frac \delta 6,
	\]which implies $\sup_{\astr \in \Str_{T_0 - 1}} S_\astr(v) \le M - \frac \delta 6 < M$. This contradicts to the correctness of the tuple $(a^\ast_1, \dots, a^\ast_{T_0 - 1}, \Str_{T_0 - 1})$.

	Now assume $\forall \astr \in \Str\ S_\astr(v) < M$, hence, $S_{\astr^\ast}(v) < M$ and there exists such strategy $\astr \in \Str$ that $M > S_\astr(v) > S_{\astr^\ast}(v)$ by the definition of $M$. Define $\varepsilon = S_\astr(v) - S_{\astr^\ast}(v)$. Consider an integer $T_1 \ge 0$ such that $\sum_{t > T_1} \g_t v < \varepsilon$ ($T_1$ exists, since the series $\sum_{t = 1}^\infty \g_t$ converges). By the definition of $\astr^\ast$ and $\Str_{T_1}$ there exists such strategy $\bstr = \{b_t\}_{t = 1}^\infty \in \Str_{T_1}$ that $S_\bstr(v) > S_\astr(v)$ (and $b_1 \dots b_{T_1} = a^\ast_1 \dots a^\ast_{T_1}$). Hence, denoting the price sequence set by $\A$ in response to $\bstr$ by $\{p^1_t\}_{t = 1}^\infty$ ($p^1_1\dots p^1_{T_1 + 1} = p_1 \dots p_{T_1 + 1}$) and using that $\sum_{t > T_1} \g_t a^\ast_t(v - p_t) \ge 0$ we gain
	\[
	\varepsilon < S_\bstr(v) - S_{\astr^\ast}(v) = \sum_{t > T_1} \g_t b_t (v - p^1_t) - \sum_{t > T_1} \g_t a^\ast_t (v - p_t) < \varepsilon - 0 = \varepsilon,
	\]which is the contradiction. Thus, the desired result $\exists \astr \in \Str\ S_\astr(v) = M$ is obtained. 
\end{proof}

\subsection{Existence of a global maximum point of $H_D(v)$ (from Section~\ref{sec_ConstAlg})}
\label{app_subsec_globmax_H_D}

Assume the non-negative random variable $V \sim D$ has a finite expectation, and the distribution function $G(v)=\Prob_{V\sim D}[v < V]$ is continuous. 
\begin{proposition}
	\label{appx_prop_globMaxH}
	In this case, the function $H_D(v) = G(v) \cdot v$ has a global maxima point. 
\end{proposition}

\begin{proof}[Proof of Proposition~\ref{appx_prop_globMaxH}]
	We denote the probability measure function by $\mu_V:\mathfrak{B}(\mathbb{R}) \rightarrow [0; 1],\ \mu_V(A) = \Prob[V \in A]$, since we already use $\Prob$ in traditional manner (e.g. $\Prob[V \ge 0]$). Here $\mathfrak{B}(\mathbb{R})$ is the Borel Algebra for $\mathbb{R}$ with the standard topology set. In this terms, the function $H_D$ can be expressed as follows
	\[	
	H_D(\bar v) = G(\bar v) \cdot \bar v = \Prob[V > \bar v] \cdot \bar v = \bar v \cdot  \int_{(\bar v; +\infty)} 1 d \mu_V,
	\]which can be upper bound by $\int_{(\bar v; +\infty)} v d \mu_V$, where $\int_A f(v) d \mu_V$ denotes the Lebegue's integral of a function $f$ on a set $A$ w.r.t. the probability measure $\mu_V$. Due to the absolute continuity of the Lebesgue integral (see Appendix~\ref{app_subsec_abscont}), the fact that $\mu_V((\bar v; +\infty)) \xrightarrow[\bar v \to \infty]{} 0$ holds, and, since $V$ has a finite expectation, we obtain that $\int_{(\bar v; +\infty)} v d \mu_V \xrightarrow[\bar v \to \infty]{} 0$ and, thus, 
	\[ 
	H_D(\bar v) \le \int_{(\bar v; +\infty)} v d \mu_V \xrightarrow[\bar v \to \infty]{} 0 \Rightarrow
	H_D(\bar v) \xrightarrow[\bar v \to \infty]{} 0.
	\] Hence, for an arbitrary picked point $v_0 > 0$ such that $H_D(v_0) > 0$ there exists such $v_1 > v_0$ that $\forall v > v_1\ H_D(v) < H_D(v_0)$, thus, if $H_D$ has a maxima point $v^\ast$ in the segment $[0; v_1]$, then $v^\ast$ is the global maxima point, but $H_D$ is a continuous function, hence, it has a maxima point in any segment, thus, Proposition~2 is proved.
	
\end{proof}

\subsection{Proof of the claims in Lemma~\ref{Rlemma} (from Section~\ref{sec_EqDiscounts})}
\label{app_subsec_proof_RLemma}

\begin{proof}[Proof of Lemma~\ref{Rlemma}]
	The proof of that $R(v)$ is increasing is provided in the main text in Section~\ref{sec_EqDiscounts}.
	
	\underline{Proof of $R(0)=0$.}
	
	Now we also note that $\forall v \ge 0\ \ R(v) \ge 0$, since $R(v) =   \sum_{t = 1}^{\infty} \g_t a_t p_t$ (as a sum of non-negative terms). By the definition $S(0) = -R(0)$ and, thus,  $R(0) = - S(0) \le 0$. Therefore, $R(0)=0$.
	
	\underline{Proof of $0 \le \Expect\left[R(V)\right] < +\infty$.}
	
	Now consider $v \ge 0$ and let $\astr = \{a_t\}_{t = 1}^\infty$ be the optimal strategy for $v$. Hence, 
	\begin{gather*}
		0 \le S(v) = \sum_{t = 1}^{\infty} \g_t a_t (v - p_t) \le \Gamma v - R(v) \Rightarrow R(v) \le \Gamma v
	\end{gather*} Thus, 
	\begin{gather*}
		\forall v \ge 0:\ 0 \le R(v) \le \Gamma v \Rightarrow 0 \le \Expect  \left[ R(V)  \right] \le \Gamma \cdot \Expect \left[ V \right]
	\end{gather*}
\end{proof}

\subsection{Proof of  Lemma~\ref{QSlemma} (from Section~\ref{sec_EqDiscounts})}
\label{app_subsec_proof_QSlemma}

\begin{proof}[Proof of  Lemma~\ref{QSlemma}]
	Consider an arbitrary optimal strategy $\astr$ for the valuation $v$, i.e., s.t. $S(v) = S_\astr(v)$. Since $S_\astr(w) =  q_\astr w - r_\astr$ for any $w\ge0$ and $S_\astr(v) = S(v)$, we can write $S_\astr(v + \delta) = S(v) + q_\astr \delta$ and by the definition of the derivative $S(v + \delta) = S(v) + S^\prime(v) \delta + o_{\delta \to 0}(\delta)$, thus $S_\astr(v + \delta) - S(v + \delta) = (q_\astr - S^\prime(v))\delta + o_{\delta \to 0}(\delta)$, which should be not greater than zero for all possible $\delta$ since $S$ is convex.
	Hence we get $q_\astr = S^\prime$, because otherwise $(q_\astr - S^\prime(v)) \delta + o_{\delta \to 0}(\delta)$ will take both positive and negative values in a neighborhood of $0$. Finally, remind that $Q(v)=q_\astr$  since $\astr$ is optimal for $v$. 
\end{proof}

\subsection{Proof of Lemma~\ref{lemma_ExpR_to_Int_dQ_dS} (from Section~\ref{subsec_OptAlgProof})}
\label{app_subsec_proof_lemma_dQ_dS}

\begin{proof}[Proof of Lemma~\ref{lemma_ExpR_to_Int_dQ_dS}]
	First we note, that due to the absolute continuity of the Lebesgue integral (see \ref{app_subsec_abscont}) and the fact that $R(V)$ has a finite expectation we can write $\Expect\left[R(V) I_{[0; \bar v]}(V)\right] \xrightarrow[\bar v \to +\infty]{} \Expect\left[R(V)\right]$. Rewrite $\Expect\left[R(V) I_{[0; \bar v]}(V)\right]$ using Lebesgue-Stieltjes integral w.r.t. the fact that $F(v) = 1 - G(v)$:
	\begin{equation}\label{LebesgueToStieltjes}
		\Expect\left[R(V) I_{[0; \bar v]}(V)\right] = \int\limits_{[0; \bar v]} R(v) dF(v) = - \int\limits_{[0; \bar v]} R(v) d G(v)
	\end{equation}For the latter integral we use the integration by parts formula (which holds since $G$ is continuous on it's domain and $R$ is non-decreasing on it's domain) and gain
	\begin{gather} \label{integration_by_parts}
		- \int\limits_{[0; \bar v]} R(v) d G(v) = - G(v) R(v) \Bigm|_0^{\bar v}  + \int\limits_{[0; \bar v]} G(v) dR(v) =\\
		- G(v) R(v) \Bigm|_0^{\bar v}   + \int\limits_{[0; \bar v]} G(v) d(Q(v) \cdot v)  - \int\limits_{[0; \bar v]} G(v) dS(v)
	\end{gather}
	
	Since the function $G$ is continuous on $[0; \bar v]$, the Riemann-Stieltjes integral $\int\limits_0^{\bar v} G(v) d(Q(v) \cdot v)$ is defined and equals to the corresponding Lebesgues-Stieltjes integral. For the Riemann-Stieltjes integral for our conditions following identity holds
	\begin{equation}\label{NewtonLeibniz}
		\int\limits_0^{\bar v} G(v) d(Q(v) \cdot v) = \int\limits_0^{\bar v} G(v) Q(v) dv + \int\limits_0^{\bar v} G(v) \cdot v dQ(v).
	\end{equation}
	Bringing together that $R(0) = 0$ (Lemma~1) and $\lim\limits_{\bar v \to +\infty} R(\bar v) G(\bar v) = 0$ (because 
	$$R(\bar v) G(\bar v) \le \int_{(\bar v; +\infty)} R(v) dF(v) \xrightarrow[\bar v \to +\infty]{} 0,$$
	 which holds due to the absolute continuity of the Lebesgue integral), Eq.~(\ref{LebesgueToStieltjes}), Eq.~(\ref{integration_by_parts}) and Eq.~(\ref{NewtonLeibniz}) (and taking the limits) we get 
	\[ 
	\Expect\left[R(V)\right] = \lim\limits_{\bar v \to +\infty} \Expect\left[R(V) I_{[0; \bar v]}(V)\right]  = \int\limits_0^{+\infty} G(v) Q(v) dv + \int\limits_0^{+\infty} G(v) v dQ(v) - \int\limits_{[0; +\infty)} G(v) dS(v).
	\]Considering first two integrals in the latter expression as Lebesgue and Lebesgue-Stieltjes integrals respectively, we obtain the desired result
	
\end{proof}

\subsection{Proof of $\wb_\gb (\A) \in \Delta^{k(T)} \Leftrightarrow \A \in \tilde \Alg(\gb)$ from Lemma~\ref{lemma_linear_transormation} (from Section~\ref{subsec_finite_games_theory})}
\label{app_subsec_proof_lemma_linear_transormation}

	The proof of the statement that $\wb_\gb (\A) \in \Delta^{k(T)}$ if and only if $\A \in \tilde \Alg(\gb)$  could be made via two following inductions.
\begin{itemize}
	\item Let $\A \in \tilde \Alg(\gb)$ and  $\vb = \wb_\gb(\A)$. 
	Then, for $ j = 1, \dots, k$, $ v_j \ge 0$ (\emph{the basis of the induction}). Indeed, assume that this condition is violated for some $j$, then  $S_{\astr^{j - 1}}(v) < S_{\astr^j}(v)$ $\fa v > v_j$, but $v_j < 0$, and, thus $\astr^{j - 1}$ is not active, which is a contradiction. 
	So, let us set $v_0:=0$ (for the notation simplicity); assume, for $s\ge0$, $0 \le v_1 \le \dots \le v_s $ and $ v_s \le v_{s + 1}, \dots, v_{k}$; and prove that $v_{s + 1} \le v_{s + 2}, \dots, v_k$ (\emph{the inductive step}).  
	
	Assume the contrary: for some $j > s + 1$ we have $v_s \le v_j < v_{s + 1}$. 
	Then \mbox{$S_{\astr^s}(v_j) > S_{\astr^{s + 1}}(v_j)$} since $S_{\astr^s}(v_{s + 1}) = S_{\astr^{s + 1}}(v_{s + 1})$ and the slope of $S_{\astr^s}$ is less than that of $S_{\astr^{s + 1}}$. If \mbox{$S_{\astr^{j}}(v_j) \ge S_{\astr^s}(v_j)$}, we have $S_{\astr^{s + 1}}(v) < S_{\astr^{j - 1}}(v)$ for $v > v_j$ and $S_{\astr^{s + 1}}(v) < S_{\astr^s}(v)$ for $v \le v_j$, which means that $\astr^{s + 1}$ is not active.
	Otherwise, we have $S_{\astr^{j - 1}}(v) < S_{\astr^s}(v)$ for $v \le v_j$ and $S_{\astr^{j - 1}}(v) < S_{\astr^{j}}(v)$ for $v > v_j$, which means that $\astr^{j - 1}$ is not active. Both cases infer contradiction, thus, the induction holds. 
	
	\item Conversely, let $\vb  = \wb_\gb(\A) \in \Delta^{k(T)}$. 
	Then $S_{\astr^j}(0) \le S_{\astr_{j - 1}}(0)$ for all $j > 0$ (and, thus, $S_{\astr^0}$ is active). 
	Indeed, assume that this condition is violated for some $j$, then $v_j < 0$, contradiction (\emph{the basis of the induction}). 
	So, let us set $v_0:=0$ (for the notation simplicity); assume, for $s\ge0$,  that 
	$$
	S_{\astr^j},j \le s, \:\hbox{are active}\:, \quad S_{\astr^j}(v_s) \le S_{\astr^{j - 1}}(v_s) \:\hbox{for}\: j> s  \quad \hbox{and}\quad S_{\astr^j}(v_s) \le S_{\astr^s}(v_s) \:\hbox{for}\: j < s;
	$$
	and prove that 
	$$
	S_{\astr^{s + 1}} \:\hbox{is active}\:, S_{\astr^j}(v_{s + 1}) \le S_{\astr^{j - 1}}(v_{s + 1}) \:\hbox{for}\: j > s + 1 \quad \hbox{and}\quad S_{\astr^j}(v_{s + 1}) \le S_{\astr^{s + 1}}(v_{s + 1})\:\hbox{for}\:  j < s + 1;
	$$
	i.e., (\emph{the inductive step}).
	
	The second condition is due to $v_j \ge v_{s + 1}$ for $j > s + 1$. The third condition for $j = s$ follows from the definition of $v_{s + 1}$ and the same for $j < s$ is due to the fact that $S_{\astr^j}(v_{s + 1}) \le S_{\astr^s}(v_{s + 1})( = S_{\astr_{s + 1}}(v_{s + 1}))$, since the slope of the function $S_{\astr^j}$ is less than the slope of the function $S_{\astr^s}$ and $S_{\astr^j}(v_s) \le S_{\astr^s}(v_s)$. The second condition together with the third conditions gives the activeness of $S_{\astr^{s + 1}}$. Thus, the induction holds. 
\end{itemize}

\subsection{The dimension reduction in the case of finite game with horizon $T = 2$} \label{app_subsec_dimension_reduction}

Let $\gb^\ttB = \{\g_b^{t - 1} \cdot I_{\{t \le 2\}}\}_{t = 1}^\infty$ and $\gb^\ttS =  \{\g_s^{t - 1} \cdot I_{\{t \le 2\}}\}_{t = 1}^\infty$  for $0 < \g_b < \g_s < 1$. Build the $\Xi_T(\gb^\ttB, \gb^\ttS)$ matrix by the definition: 
\[
\Xi_T(\gb^\ttB, \gb^\ttS) = \begin{pmatrix}
\gamma_s & 0 & 0 \\ -(\gamma_s - \gamma_b) & 1 - \gamma_b & 0 \\ 0 & 0 & \gamma_s
\end{pmatrix}
\] Now we prove that $L(v) = (1 - F(v))^\intercal \Xi_T(\gb^\ttB, \gb^\ttS) v$ always has a maximum on the hyperplane $v_2 = v_3$. 

 Indeed, assume it's not and consider a maximum point with coordinates $v_1, v_2, v_3$ such that $v_2 < v_3$.   Consider two possible cases: $(1 - F(v_3)) v_3 > (1 - F(v_2)) v_2$ and $(1 - F(v_3)) v_3 \le (1 - F(v_2)) v_2$. 
 \begin{itemize}
 	\item   $(1 - F(v_3)) v_3 > (1 - F(v_2)) v_2$ implies $L(v_1, v_2, v_3) < L(v_1, v_3, v_3)$, since 
\[
L(v_1, v_3, v_3) - L(v_1, v_2, v_3) = -(\gamma_s - \gamma_b) (F(v_2) - F(v_3)) v_1 + (1 - \gamma_b) ((1 - F(v_3)) v_3 - (1 - F(v_2)) v_2),  
\]where the left term is non-negative, since $v_3 > v_2$, and the right term is strictly positive by the assumption. $L(v_1, v_2, v_3) < L(v_1, v_3, v_3)$ contradicts to our assumption.

\item $(1 - F(v_3)) v_3 \le (1 - F(v_2)) v_2$ implies $L(v_1, v_2, v_3) \le L(v_1, v_2, v_2)$, since 
\[
L(v_1, v_2, v_2) - L(v_1, v_2, v_3) = \gamma_s ((1 - F(v_2)) v_2 - (1 - F(v_3)) v_3) \ge 0.
\]This also contradicts to our assumption.
\end{itemize}

Both possible cases infer contradiction, thus, our assumption is wrong. Q.E.D.

\subsection{Proof of  Proposition~\ref{prop_approx_by_finite_alg} (from Section~\ref{subsec_infinite_game_studies})}
\label{app_subsec_proof_prop_approx_by_finite_alg}

\begin{proof}[Proof of Proposition~\ref{prop_approx_by_finite_alg}]
	The left inequality is trivial, since $\Alg_\tau \subset \Alg$. The second obvious observation is that $\SRev_{\gb^\ttS, \gb^\ttB}(\A, v)$ for $\A \in \Alg_\tau$ is equal to $\SRev_{\gb^\ttS, \tilde \gb^\ttB}(\A, v)$, where $\tilde \gb^\ttS = \{\tilde \g_t\}_{t = 1}^\infty$ is equal to $\gb^\ttS$ on $[1; \tau - 1]$, $\tilde \g_\tau = \sum_{t = \tau}^{\infty} \g^\ttS_t$ and $\tilde \g_t = 0$ for $t > \tau$. For such sellers discount $\tilde \gb^\ttS$ the strategic revenue is obviously does not depend on algorithm prices for rounds $\tau + 1, \tau + 2, \dots$, thus, 
	\[
	\max_{\A \in \Alg_\tau} \Expect\left[\SRev_{\gb^\ttS, \gb^\ttB}(\A, V)\right]  = \max_{\A \in \Alg} \Expect\left[\SRev_{\gb^\ttS, \tilde \gb^\ttB}(\A, V)\right].
	\]The following step of the proof is formulated as a lemma:
	\begin{lemma} \label{lemma_discounts_sum_bound}
		Let $\gb^\ttS, \gb^{\ttS, 1}, \gb^{\ttS, 2}$ be discounts such that $\gb^\ttS = \gb^{\ttS, 1} + \gb^{S, 2}$. In this case 
		\[
		\max_{\A \in \Alg} \Expect\left[\SRev_{\gb^\ttS, \gb^\ttB}(\A, V)\right] \le 	\max_{\A \in \Alg} \Expect\left[\SRev_{\gb^{\ttS, 1}, \gb^\ttB}(\A, V)\right] + 
		\max_{\A \in \Alg} \Expect\left[\SRev_{\gb^{\ttS, 2}, \gb^\ttB}(\A, V)\right]
		\]
	\end{lemma}

	We omit the proof, since it is trivial. Apply Lemma \ref{lemma_discounts_sum_bound} to the sellers discount divided into two parts as follows: \mbox{$\gb^\ttS = \gb^\ttS \cdot I_{\{t > \tau\}} + \gb^\ttS \cdot I_{\{t \le \tau\}}$}: 
	\[
	\max_{\A \in \Alg} \Expect\left[\SRev_{\gb^\ttS, \gb^\ttB}(\A, V)\right] \le 	\max_{\A \in \Alg} \Expect\left[\SRev_{ \gb^\ttS \cdot I_{\{t \le \tau\}}, \gb^\ttB}(\A, V)\right] + 
	\max_{\A \in \Alg} \Expect\left[\SRev_{ \gb^\ttS \cdot I_{\{t > \tau\}} , \gb^\ttB}(\A, V)\right]
	\]
	The left term of the right-hand side of the inequality is not greater than $\max_{\A \in \Alg} \Expect\left[\SRev_{\gb^\ttS, \tilde \gb^\ttB}(\A, V)\right]$, since $\tilde \gb^\ttS \ge \gb^\ttS \cdot I_{\{t \le \tau\}}$ and the right term is not greater than $\Gamma^\ttS_\tau\Expect\left[V\right]$. This fact can be proved in following several steps: 
	\begin{enumerate}
		\item Following identitiy can be verifyed by the direct application of $\nu(\gb)$ definition: 
		\[
			\g^\ttS_{\tau + i} = \g^\ttB_{\tau + i} \cdot \frac{\nu(\gb^\ttS)_{\tau + i - 1}}{\nu(\gb^\ttB)_{\tau + i - 1}} \cdot \dots \cdot \frac{\nu(\gb^\ttS)_{\tau + 1}}{\nu(\gb^\ttB)_{\tau + 1}} \cdot \frac{\g^\ttS_{\tau}}{\g^\ttB_{\tau}}
		\]
		\item Define $c_{i} := \frac{\nu(\gb^\ttS)_{\tau + i - 1}}{\nu(\gb^\ttB)_{\tau + i - 1}} \cdot \dots \cdot \frac{\nu(\gb^\ttS)_{\tau + 1}}{\nu(\gb^\ttB)_{\tau + 1}} \cdot \frac{\g^\ttS_{\tau}}{\g^\ttB_{\tau}}$ for $i \ge 1$, thus, $c_i$ are increasing and 
		\[
			 \g^\ttS_{\tau + i} = c_i \cdot \g^\ttB_{\tau + i}
		\]
		\item In this case 
		\[
			\gb^\ttS \cdot \Ind_{\{t > \tau\}} = c_1 \gb^\ttB \cdot  \Ind_{\{t > \tau\}} + (c_2 - c_1) \gb^\ttB \cdot  \Ind_{\{t > \tau + 1\}} + (c_3 - c_2) \gb^\ttB \cdot  \Ind_{\{t > \tau + 2\}} 
		\]
		\item Consider a case when the discount of the seller is $c \cdot \gb^\ttB \cdot \Ind_{\{t > \tau + i\}}$ for some $c > 0$. Let $\A \in \Alg$ and $\astr = \{a_t\}_{t = 1}^\infty$ be some optimal strategy for a valuation $v > 0$. Then 
		\[
			S(v) = \sum_{t = 1}^{\infty} a_t \g^\ttB_t (v - \A(\astr_{1:t - 1})) \ge \sum_{t = 1}^{\tau + i} a_t \g^\ttB_t (v - \A(\astr_{1:t - 1})) \Rightarrow \sum_{t = \tau + i + 1}^{\infty} a_t \g^\ttB_t v \ge \sum_{t = \tau + i + 1}^{\infty} a_t \g^\ttB_t \A(\astr_{1:t - 1}) 
		\]
		But the right part of the last inequality is $\frac 1c \SRev_{c \gb^\ttB \cdot \Ind_{\{t > \tau + i\}}, \gb^\ttB }(\A, v)$, and, thus, 
		\[
			c \Gamma^B_{\tau + i} v \ge  \SRev_{ c \gb^\ttB \cdot \Ind_{\{t > \tau + i\}}, \gb^\ttB}(\A, v) \Rightarrow c \Gamma^B_{\tau + i} \Expect\left[V\right] \ge \max_{\A \in \Alg}\Expect\left[\SRev_{c \gb^\ttB \cdot \Ind_{\{t > \tau + i\}}, \gb^\ttB}(\A, V)\right]
		\]
		\item Finally, apply Lemma \ref{lemma_discounts_sum_bound} and the identity from our third step and get (for the notation simplicity $c_0 := 0$): 
		\[
			 \max_{\A \in \Alg} \Expect\left[\SRev_{ \gb^\ttS \cdot I_{\{t > \tau\}} , \gb^\ttB}(\A, V)\right] \le \sum_{i = 1}^{\infty} (c_i - c_{i - 1}) \Gamma^\ttB_{\tau + i - 1} \Expect\left[V\right]  = \Gamma^\ttS_{\tau} \Expect\left[V\right]  \qquad \hbox{Q.E.D.}
		\]
		
	\end{enumerate}
	
\end{proof}

\section{Auxiliary definitions and propositions}

\subsection{Absolute continuity of the Lebesgue integral (\cite{kolmogorov2012introductory})} 
\label{app_subsec_abscont}

\begin{proposition}
	Consider the Lebesgue measure $\mu$ on $\mathbb{R}$ and let $f:\mathbb{R} \rightarrow \mathbb{R}$ be an integrable function on $A \subset \mathbb{R}$, then for any $\varepsilon > 0$ there exists such $\delta > 0$ that 
	\[
	\left|\int_B f(x) d\mu \right| < \varepsilon,
	\]where $B \subset A: \mu(B) < \delta$. 
\end{proposition}

\section{Numerical Solutions for different distributions}
\label{app_sec_numsulutions_diff_distr}

Here we provide plots for $V \sim \beta(4, 2), V \sim \beta(2, 4)$, and $V$ distributed with the density \mbox{$\frac{1 - e^{-x}}{1 - e^{-1}} \Ind_{0 \le x \le  1}$} in different special cases. Discount are taken identically to those from Section \ref{subsec_finite_game_studies} and Section \ref{subsec_infinite_game_studies} (as well as grids for $\g_\ttB$ and $\g_\ttS$). Figures descriptions are given in the following list: 
\begin{enumerate}
	\item Figure \ref{img_T2_Beta42}: 2-round game for $V \sim \beta(4, 2)$.
	\item Figure \ref{img_T2_Beta24}: 2-round game for $V \sim \beta(2, 4)$.
	\item Figure \ref{img_T3_Beta42}: 3-round game for $V \sim \beta(4, 2)$.
	\item Figure \ref{img_T3_Beta24}: 3-round game for $V \sim \beta(2, 4)$.
	\item Figure \ref{img_T4_ExpInf}: infinite game for $V$ distributed with the density \mbox{$\frac{1 - e^{-x}}{1 - e^{-1}} \Ind_{0 \le x \le  1}$}, which is a density of a random variable $\xi \sim Exp(1)$ conditioned by $0 \le \xi \le 1$. 
\end{enumerate} 

\begin{figure}
	\centering
	\vspace{-2mm}
	\includegraphics[width=\columnwidth]{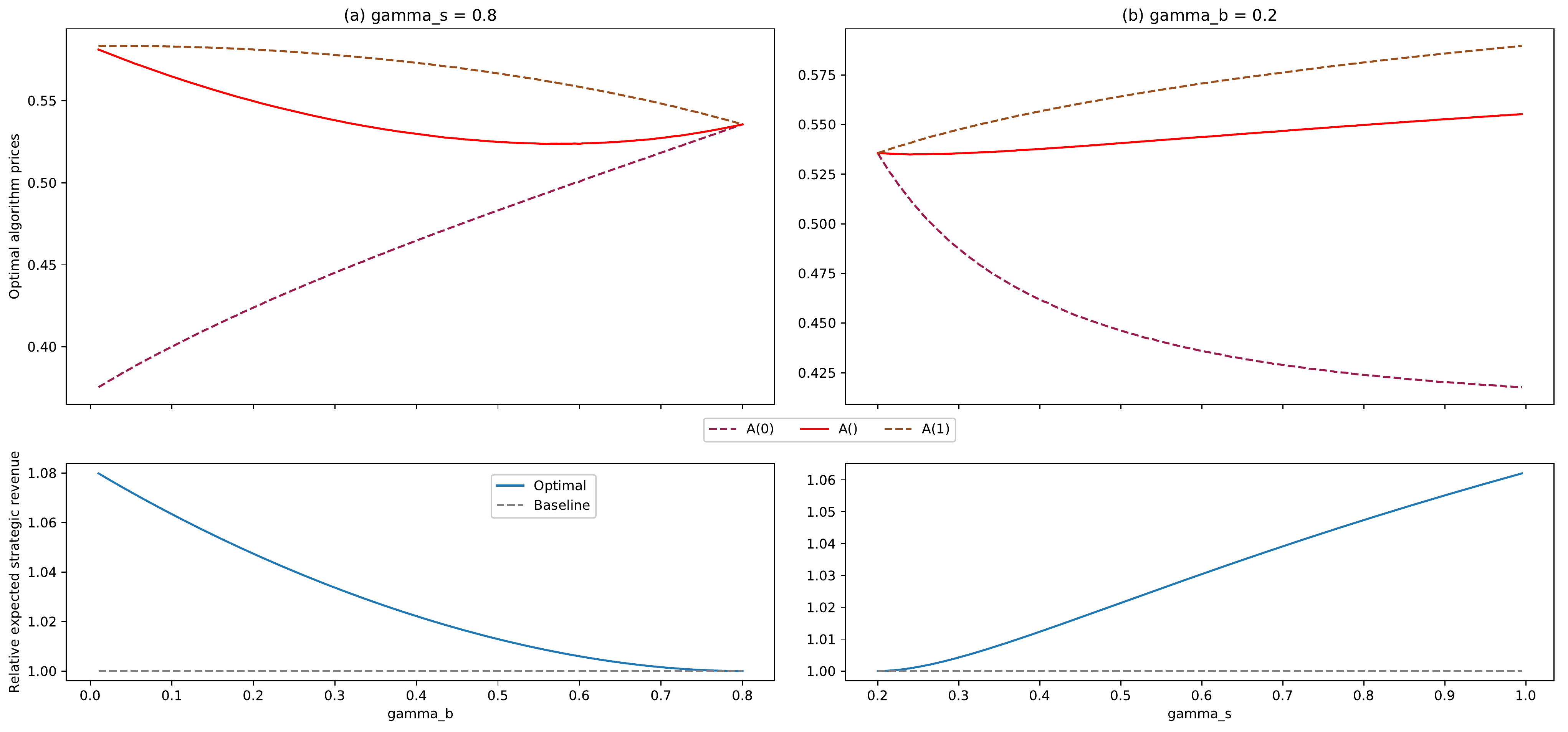}
	\vspace{-8mm}
	\caption{$2$-round game. The prices  $\A^\ast(\nO), \A^\ast(\estr), \A^\ast(\nI)$ and the relative expected strategic revenue (w.r.t.\ $\A^*_D$) of the optimal algorithm $\A^\ast$ for discount rates:
		(a) $\g_\ttS = 0.8$ and various $\g_\ttB$;
		(b) $\g_\ttB = 0.2$ and various $\g_\ttS$.}
	\label{img_T2_Beta42}
	\vspace{-4mm}
\end{figure}

\begin{figure}
	\centering
	\vspace{-2mm}
	\includegraphics[width=\columnwidth]{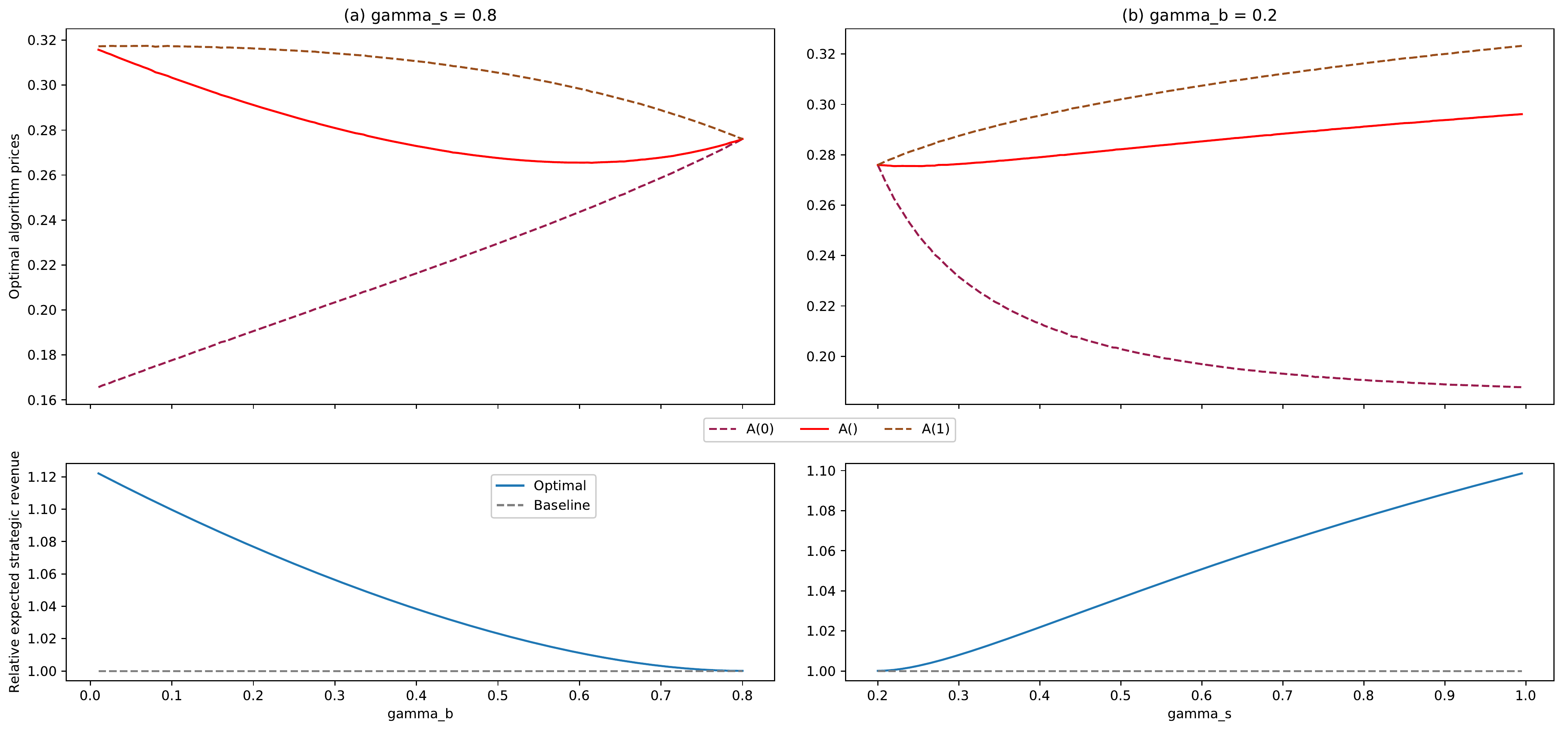}
	\vspace{-8mm}
	\caption{$2$-round game. The prices  $\A^\ast(\nO), \A^\ast(\estr), \A^\ast(\nI)$ and the relative expected strategic revenue (w.r.t.\ $\A^*_D$) of the optimal algorithm $\A^\ast$ for discount rates:
		(a) $\g_\ttS = 0.8$ and various $\g_\ttB$;
		(b) $\g_\ttB = 0.2$ and various $\g_\ttS$.}
	\label{img_T2_Beta24}
	\vspace{-4mm}
\end{figure}

\begin{figure}
	\centering
	\vspace{-2mm}
	\includegraphics[width=\columnwidth]{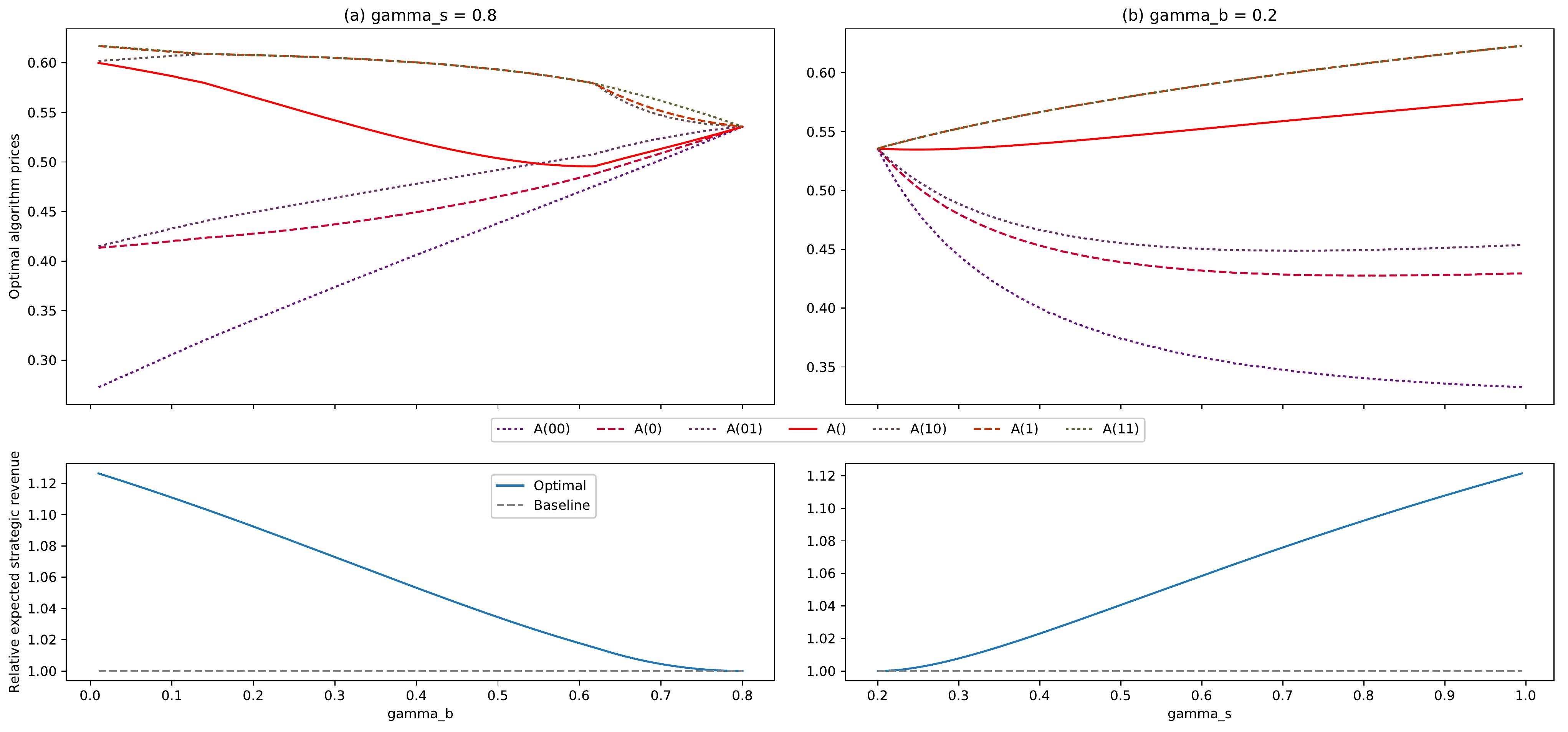}
	\vspace{-8mm}
	\caption{$3$-round game. The prices  $\A^\ast(\n)$, for nodes $\n\in\N$ s.t.\  $|\n|\le 2$, and relative expected strategic revenue (w.r.t.\ $\A^*_D$) of the optimal algorithm $\A^\ast$ for discounts:
		(a) $\g_\ttS = 0.8$ and various $\g_\ttB$;
		(b)  $\g_\ttB = 0.2$ and various $\g_\ttS$.}
	\label{img_T3_Beta42}
	\vspace{-4mm}
\end{figure}

\begin{figure}
	\centering
	\vspace{-2mm}
	\includegraphics[width=\columnwidth]{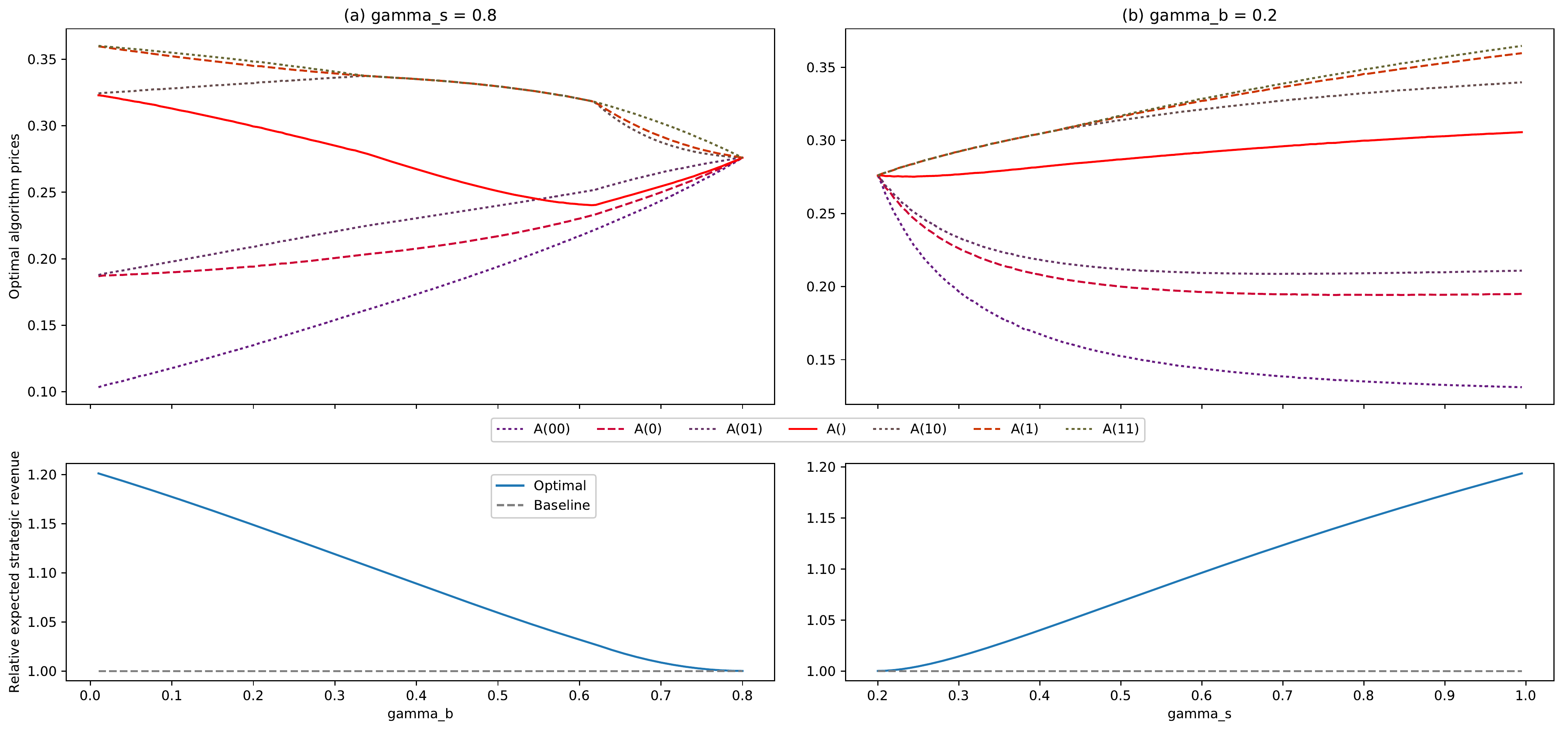}
	\vspace{-8mm}
	\caption{$3$-round game. The prices  $\A^\ast(\n)$, for nodes $\n\in\N$ s.t.\  $|\n|\le 2$, and relative expected strategic revenue (w.r.t.\ $\A^*_D$) of the optimal algorithm $\A^\ast$ for discounts:
		(a) $\g_\ttS = 0.8$ and various $\g_\ttB$;
		(b)  $\g_\ttB = 0.2$ and various $\g_\ttS$.}
	\label{img_T3_Beta24}
	\vspace{-4mm}
\end{figure}

\begin{figure}
	\centering
	\vspace{-2mm}
	\includegraphics[width=\columnwidth]{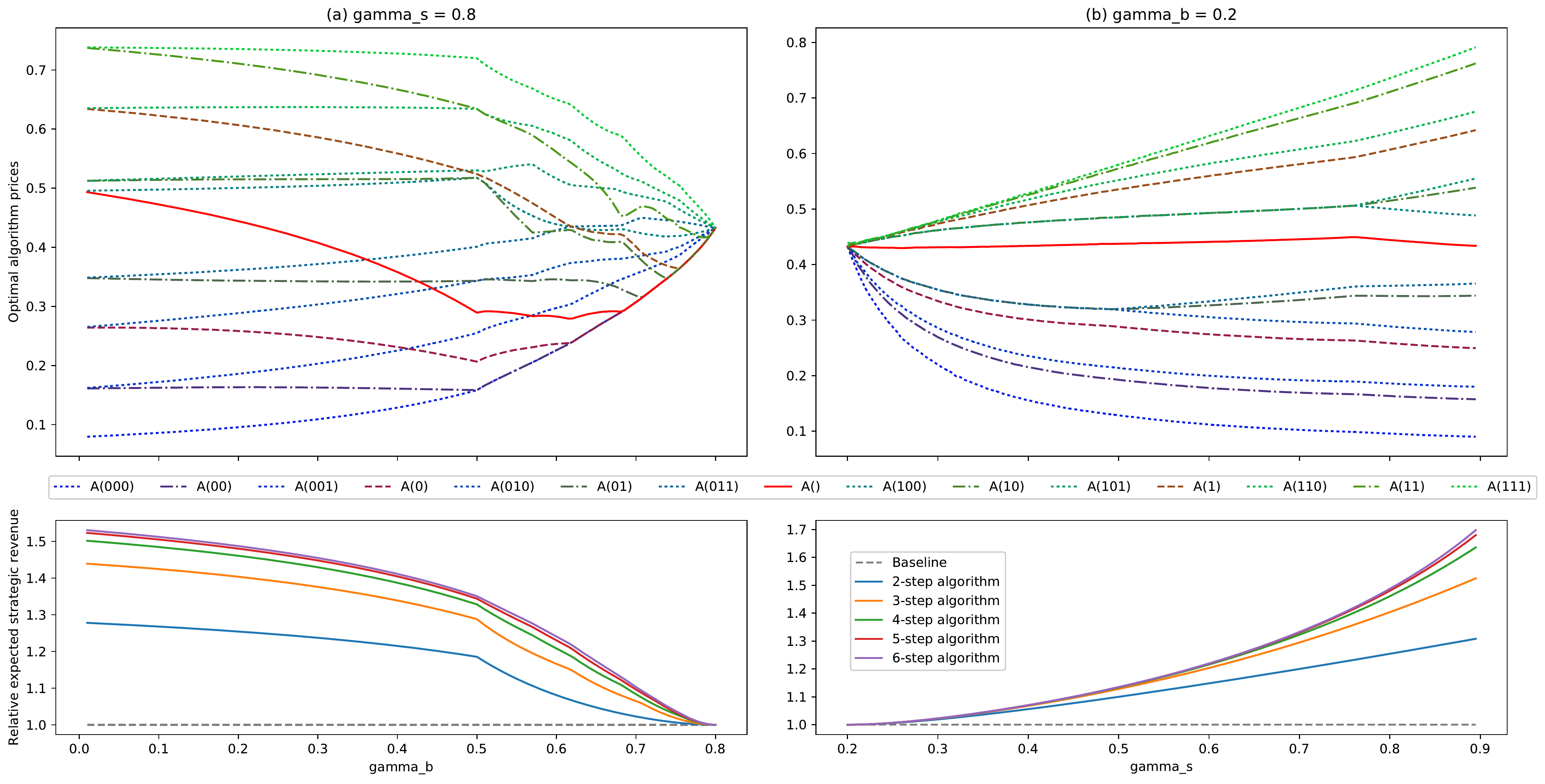}
	\vspace{-8mm}
	\caption{Infinite game. The prices  $\A^\ast_4(\n)$, for nodes $\n\in\N$ s.t.\  $|\n|\le 3$, of the optimal $4$-step algorithm $\A^\ast_4$ and the relative expected strategic revenue (w.r.t.\ $\A^*_D$) of the optimal $\tau$-step algorithm $\A^\ast_\tau, \tau=2,..,6,$ for discounts:
		(a) $\g_\ttS = 0.8$ and various $\g_\ttB$;
		(b)  $\g_\ttB = 0.2$ and various $\g_\ttS$.}
	\label{img_T4_ExpInf}
	\vspace{-4mm}
\end{figure}

\section{The analog of the Envelope theorem}
Here we show how one can use the abstract method to obtain results in case $\gb^\ttS \le \gb^\ttB$. First of all let us introduce the generalization of the envelope theorem: 

\begin{theorem}
	Let $V = I_1 \times \dots \times I_{|A|}$, where $I_j$ is an interval or a segment.\footnote{Think of $V$ as a type of a buyer: if the set of alternatives in a game is $A$, then for each alternative $i$-th buyer has a valuation $v_i$, the vector $(v_1, \dots, v_{|A|})$ is the type of the buyer.} Let $p, \omega: V_i \to \mathbb{R}, \mathbb{R}^{|A|}$\footnote{$\omega$ corresponds to the distribution of alternatives, $p$ corresponds to the payment.} such that 
	\begin{enumerate}
		\item $\sum_{i = 1}^{|A|} \omega_i(v) = 1$ for all $v \in V$.
		\item $\forall v_0, v_1 \hspace{0.5cm} v_0 \cdot \omega(v_0) - p(v_0) \ge v_0 \cdot \omega(v_1) - p(v_1)$ (where $a \cdot b$ is for a scalar product).\footnote{This property means that the optimal buyers strategy is to tell true about his type in the mechanism $(p, \omega)$.}
	\end{enumerate}
	Then following statements about $\omega$ and $p$ are true: 
	\begin{enumerate}
		\item $U(v) := v \cdot \omega(v) - p(v)$ is convex, $\omega_i, p$ are increasing in $v_i$ for all $i$. 
		\item The differential $d U(v)$ is defined almost everywhere and equals to $\omega(v)$. 
	\end{enumerate}
\end{theorem}

\begin{proof}
	First let us prove the convexity of $U$. Consider $v_0, v_1 \in V$. Denote $v_a = (1 - a) v_0 + a v_1$ for $a \in [0; 1]$. Then 
	\begin{align*}
		U(v_a) &= v_a \cdot \omega(v_a) - p(v_a) = (1 - a) (v_0 \cdot \omega(v_a) - p(v_a)) + a (v_1 \cdot \omega(v_a) - p(v_a)) \le \\
		&\le  (1 - a) (v_0 \cdot \omega(v_0) - p(v_0)) + a (v_1 \cdot \omega(v_1) - p(v_1)) = (1 - a) U(v_0) + a U(v_1),  
	\end{align*}thus, by the definition $U$ is convex. 
	
	Consider $v = (v_i, v_{-i})$ and $v^\prime = (v_i^\prime, v_{-i})$ for $v_i^\prime > v_i$. Then by the second property
	\begin{gather*}
		 v_i \omega_i(v) - p(v) \ge v_i \omega_i(v^\prime) - p(v^\prime) \\ 
		 v_i^\prime \omega_i(v^\prime) - p(v^\prime) \ge v_i^\prime \omega_i(v) - p(v)
	\end{gather*}Rearranging these inequalities we get 
	\[
		v_i^\prime(\omega_i(v) - \omega_i(v^\prime)) \le p(v) - p(v^\prime) \le v_i (\omega_i(v) - \omega_i(v^\prime)), 
	\]which can be satisfied only if $\omega_i(v) - \omega_i(v^\prime) \le 0$, since $v_i^\prime > v_i$. Thus, $\omega_i(v) \le \omega_i(v^\prime)$ and $p(v) \le p(v^\prime)$, i.e. $\omega_i, p$ are increasing in $v_i$ for all $i$, since $i, v_{-i}, v_i < v_i^\prime$ were chosen arbitrarily. 
	
	Finally let $v \in V$ and $\delta v$ be the increment, in this case by the second property
	\[
		\delta v \cdot \omega (v) \le U(v + \delta v) - U(v) \ge \delta v \cdot  \omega(v + \delta v), 
	\]which infers that $d U(v)$ equals to $\omega(v)$, if $\omega$ is continuous in $v$, which is almost everywhere. \emph{Q.E.D.}
\end{proof}

\end{document}